\documentclass[sigconf,nonacm]{acmart}

\usepackage{ulem}
\usepackage{tikz}
\usepackage{amsmath}
\usepackage{xspace}
\usepackage{enumitem}
\usepackage{cleveref}
\usepackage{caption}
\usepackage{subcaption}
\usepackage{algorithmicx}
\usepackage{thmtools}
\usepackage{thm-restate}
\usepackage[most]{tcolorbox}

\definecolor{mygreen}{RGB}{20,140,80}
\definecolor{linkcolor}{RGB}{100,0,0}
\definecolor{mylightgray}{RGB}{230,230,230}
\definecolor{verylightgray}{RGB}{245,245,245}

\algnewcommand{\IIf}[1]{\State\algorithmicif\ #1\ \algorithmicthen}
\algnewcommand{\EndIIf}{\unskip\ \algorithmicend\ \algorithmicif}

\newcounter{myalgctr}

\newtcolorbox{OuterBox}[1][]{%
    breakable,
    enhanced,
    frame hidden,
    interior hidden,
    left=-5pt,
    right=-5pt,
    top=-5pt,
    float=p,
    boxsep=0pt,
    arc=0pt
#1}%

\newtcolorbox{InnerBox}[1][]{%
    enforce breakable,
    enhanced,
    colback=gray,
    colframe=white,
#1}%

\newenvironment{tbox}{
\vspace{0.2cm}
\begin{tcolorbox}[
                  enhanced,
                  boxsep=2pt,
                  left=1pt,
                  right=1pt,
                  top=4pt,
                  boxrule=1pt,
                  arc=0pt,
                  colback=white,
                  colframe=black,
	              breakable
                  ]
}{
\end{tcolorbox}
}

\newcommand{\tboxhrule}[0]{\vspace{0.1cm} {\color{black} \hrule} \vspace{0.2cm}}
\newenvironment{titledtbox}[1]{\begin{tbox}#1 \tboxhrule}{\end{tbox}}
\newenvironment{tboxalg}[1]{\refstepcounter{myalgctr}\begin{titledtbox}{\textbf{Algorithm \themyalgctr.} #1}}{\end{titledtbox}}

\newcommand\vldbdoi{XX.XX/XXX.XX}
\newcommand\vldbpages{XXX-XXX}
\newcommand\vldbvolume{16}
\newcommand\vldbissue{1}
\newcommand\vldbyear{2022}
\newcommand\vldbauthors{\authors}
\newcommand\vldbtitle{\shorttitle}
\newcommand\vldbavailabilityurl{https://github.com/cmuparlay/PIM-tree}
\newcommand\vldbpagestyle{plain}
\setcopyright{none}
\usepackage[font=small,aboveskip=0pt,belowskip=1pt]{caption}
\ifdefined\confversion
\usepackage[compact]{titlesec}
 \titlespacing*{\section}
 {0pt}{.7ex plus .5ex minus .2ex}{.3ex plus .1ex}
 \titlespacing*{\subsection}
 {0pt}{.5ex plus .5ex minus .2ex}{.3ex plus .1ex}
\fi

\begin{document}

 \newcommand{\yan}[1]{{\color{magenta} {\bf Yan:} #1}}

\newcommand{\hongbo}[1]{}
\newcommand{\charlie}[1]{}
\newcommand{\laxman}[1]{}
\newcommand{\guy}[1]{}
\newcommand{\guyup}[1]{}
\newcommand{\phil}[1]{}
\newcommand{\yiwei}[1]{}

\definecolor{forestgreen}{rgb}{0.13, 0.55, 0.13}
\newcommand{\forestgreen}[1]{{\color{forestgreen}{#1}}}

\ifdefined\confversion
\newcommand{\revision}[1]{{\color{forestgreen} #1}}
\newcommand{\revisionweak}[1]{{\color{teal} #1}}
\else
\newcommand{\revision}[1]{{#1}}
\newcommand{\revisionweak}[1]{{#1}}
\fi

\newcommand{\conffulldifferent}[2]{\ifdefined\confversion#1\else{#2}\fi}

\newcommand{\confversiononly}[1]{\conffulldifferent{#1}{}}
\newcommand{\fullversiononly}[1]{\conffulldifferent{}{#1}}

\newcommand{\pimtree}{PIM-tree}
\newcommand{\shadowtree}{Shadow tree\xspace}
\newcommand{\pimbalance}{PIM-balance\xspace}
\newcommand{\pimbalanced}{PIM-balanced\xspace}
\newcommand{\pimimbalance}{PIM-imbalance\xspace}
\newcommand{\pimimbalanced}{PIM-imbalanced\xspace}

\newcommand{\naive}{na\"ive search}

\newcommand{\mb}[1]{{\mbox{\textit{#1}}}}
\newcommand{\smb}[1]{{\scriptsize \mbox{\textit{#1}}}}
\newcommand{\mf}[1]{{\mbox{\sc{#1}}}}
\newcommand{\whp}{\textit{whp}\xspace}

\newcommand{\note}[1]{{\textcolor{red}{Note: #1}}}
\newcommand{\polylog}[1]{\mbox{polylog($#1$)}}

\newcommand{\GET}{\textsc{Get}\xspace}
\newcommand{\UPDATE}{\textsc{Update}\xspace}
\newcommand{\UPSERT}{\textsc{Upsert}\xspace}
\newcommand{\INSERT}{\textsc{Insert}\xspace}
\newcommand{\DELETE}{\textsc{Delete}\xspace}
\newcommand{\SUCCESSOR}{\textsc{Successor}\xspace}
\newcommand{\PREDECESSOR}{\textsc{Predecessor}\xspace}
\newcommand{\SCAN}{\textsc{Scan}\xspace}
\newcommand{\RANGEQUERY}{\textsc{RangeQuery}\xspace}
\newcommand{\RANGEUPDATE}{\textsc{RangeUpdate}\xspace}
\newcommand{\RANGEOPERATION}{\textsc{RangeOperation}\xspace}
\newcommand{\BatchGET}{\textsc{BatchGet}\xspace}
\newcommand{\Function}{\textsc{Func}\xspace}

\newcommand{\skiplist}{skip list\xspace}
\newcommand{\skiplists}{skip lists\xspace}
\newcommand{\SkipList}{Skip List\xspace}
\newcommand{\SkipLists}{Skip Lists\xspace}
\newcommand{\bplustree}{B+-tree\xspace}
\newcommand{\bplustrees}{B+-trees\xspace}
\newcommand{\bst}{binary search tree\xspace}
\newcommand{\abtree}{(a,b)-tree\xspace}

\newcommand{\stoa}{state-of-the-art\xspace}

\newcommand{\addr}{\mb{addr}\xspace}

\newcommand{\LEFT}{\mb{left}\xspace}
\newcommand{\RIGHT}{\mb{right}\xspace}
\newcommand{\UP}{\mb{up}\xspace}
\newcommand{\DOWN}{\mb{down}\xspace}
\newcommand{\LOCALLEFT}{\mb{local-left}\xspace}
\newcommand{\LOCALRIGHT}{\mb{local-right}\xspace}
\newcommand{\NEXTLEAF}{\mb{next-leaf}\xspace}

\newcommand{\defn}[1]{\textit{\textbf{#1}}}
\newcommand{\mytitle}[1]{\smallskip{\bf #1.}}
\newcommand{\myparagraph}[1]{\smallskip\noindent {\bf #1}}
\newcommand{\mysubsection}[1]{\smallskip\noindent {\bf \large #1}}
\newcommand{\mystep}[1]{\smallskip\noindent {\bf #1}}

\newcommand{\tasksend}{\textsf{TaskSend}\xspace}
\newcommand{\remoteread}{\textsf{RemoteRead}\xspace}
\newcommand{\remotewrite}{\textsf{RemoteWrite}\xspace}

\newcommand{\floor}[1]{\left\lfloor #1 \right\rfloor}

\newcommand{\batch}{S}
\newcommand{\numCores}{32\xspace}
\newcommand{\numDPUs}{2048\xspace}
\newcommand{\numDPUsAll}{2048\xspace}
\newcommand{\numFanout}{16\xspace}
\newcommand{\batchSize}{1 million\xspace}
\newcommand{\scanBatchSize}{10 thousand\xspace}
\newcommand{\dataSize}{500 million\xspace}
\newcommand{\testSize}{100 million\xspace}
\newcommand{\scanTestSize}{1 million\xspace}

\newcommand{\PushPullSearch}{Push-Pull search\xspace}
\newcommand{\PushPull}{Push-Pull\xspace}
\newcommand{\Push}{Push\xspace}
\newcommand{\Pull}{Pull\xspace}
\newcommand{\PullOnly}{Pull-Only\xspace}
\newcommand{\PushOnly}{Push-Only\xspace}

\newcommand{\hide}[1]{}

\newcommand{\pimload}{{\sc PIM\_Load}\xspace}
\newcommand{\pimlaunch}{{\sc PIM\_Launch}\xspace}
\newcommand{\pimstatus}{{\sc PIM\_Status}\xspace}
\newcommand{\pimbroadcast}{{\sc PIM\_Broadcast}\xspace}
\newcommand{\pimscatter}{{\sc PIM\_Scatter}\xspace}
\newcommand{\pimgather}{{\sc PIM\_Gather}\xspace}

\newcommand{\SearchRequired}{{\sc SearchRequired}\xspace}
\newcommand{\FetchAll}{{\sc FetchAll}\xspace}

\newcommand{\githubLink}{\url{https://github.com/cmuparlay/PIM-tree}}
\newcommand{\arxivFullLink}{XXXX.XXXXX}

\confversiononly{
\title{\vspace{-0.19cm}PIM-tree: A Skew-resistant Index for Processing-in-Memory}
}
\fullversiononly{
\title{PIM-tree: A Skew-resistant Index for Processing-in-Memory}
}

\author{Hongbo Kang}
\email{khb20@mails.tsinghua.edu.cn}
\affiliation{%
  \institution{Tsinghua University}
}

\author{Yiwei Zhao}
\email{yiweiz3@andrew.cmu.edu}
\affiliation{
  \institution{Carnegie Mellon University}
}

\author{Guy E. Blelloch}
\email{guyb@cs.cmu.edu}
\affiliation{%
  \institution{Carnegie Mellon University}
}

\author{Laxman Dhulipala}
\email{laxman@umd.edu}
\affiliation{%
\institution{University of Maryland}
}

\author{Yan Gu}
\email{ygu@cs.ucr.edu}
\affiliation{%
    \institution{UC Riverside}
}

\author{Charles	McGuffey}
\email{cmcguffey@reed.edu}
\affiliation{%
  \institution{Reed College}
}

\author{Phillip B. Gibbons}
\email{gibbons@cs.cmu.edu}
\affiliation{%
  \institution{Carnegie Mellon University}
}

\begin{abstract}

The performance of today's in-memory indexes is bottlenecked by the
memory latency/bandwidth wall.  Processing-in-memory (PIM) is an
emerging approach that potentially mitigates this bottleneck, by
enabling low-latency memory access whose aggregate memory bandwidth
scales with the number of PIM nodes.  There is an inherent tension,
however, between minimizing inter-node communication and achieving
load balance in PIM systems, in the presence of workload skew.  This
paper presents \textit{PIM-tree}, an ordered index for PIM systems
that achieves both low communication and high load balance, regardless
of the degree of skew in the data and the queries.  Our skew-resistant
index is based on a novel division of labor between the multi-core
host CPU and the PIM nodes, which leverages the strengths of each.  We
introduce \textit{push-pull search}, which dynamically decides whether
to push queries to a PIM-tree node (CPU $\rightarrow$ PIM-node) or
pull the node's keys back to the CPU (PIM-node $\rightarrow$ CPU)
based on workload skew.  Combined with other PIM-friendly
optimizations (\textit{shadow subtrees} and \textit{chunked \skiplists}), our PIM-tree provides high-throughput, (guaranteed) low
communication, and (guaranteed) high load balance, for batches of
point queries, updates, and range scans.  

We implement the PIM-tree structure, in addition to prior proposed PIM
indexes, on the latest PIM system from UPMEM, with \numCores CPU cores
and \numDPUs PIM nodes.  On workloads with \dataSize keys and batches
of \batchSize queries, the throughput using PIM-trees is up to
$69.7\times$ and $59.1\times$ higher than the two best prior PIM-based methods.
As far as we know these are the first implementations of an ordered
index on a real PIM system.
\end{abstract}

\maketitle

\confversiononly{
\title{PIM-tree: A Skew-resistant Index for Processing-in-Memory}
}

\pagestyle{\vldbpagestyle}
\confversiononly{
\vspace{-.2cm}
\begingroup\small\noindent\raggedright\textbf{PVLDB Reference Format:}\\
\vldbauthors. \vldbtitle. PVLDB, \vldbvolume(\vldbissue): \vldbpages, \vldbyear.\\
\href{https://doi.org/\vldbdoi}{doi:\vldbdoi}
\endgroup
\begingroup
\renewcommand\thefootnote{}\footnote{\noindent
This work is licensed under the Creative Commons BY-NC-ND 4.0 International License. Visit \url{https://creativecommons.org/licenses/by-nc-nd/4.0/} to view a copy of this license. For any use beyond those covered by this license, obtain permission by emailing \href{mailto:info@vldb.org}{info@vldb.org}. Copyright is held by the owner/author(s). Publication rights licensed to the VLDB Endowment. \\
\raggedright Proceedings of the VLDB Endowment, Vol. \vldbvolume, No. \vldbissue ISSN 2150-8097. \\
\href{https://doi.org/\vldbdoi}{doi:\vldbdoi} \\
}\addtocounter{footnote}{-1}\endgroup

\ifdefempty{\vldbavailabilityurl}{}{
\vspace{.05cm}
\begingroup\small\noindent\raggedright\textbf{PVLDB Artifact Availability:}\\
\vspace{-.05cm}
The source code, data, and/or other artifacts have been made available at \url{\vldbavailabilityurl}.
\endgroup
}
}

\section{Introduction}
\label{sec:introduction}

The mismatch between CPU speed and memory speed (a.k.a.~the ``memory wall'') makes memory accesses the
dominant cost in today's data-intensive applications.
Traditional architectures use multi-level caches to reduce data
movement between CPUs and memory, but if the application exhibits limited locality,
most data accesses are still serviced by the memory.
This excessive data movement incurs significant energy cost,
and performance is bottlenecked by
high memory latency and/or limited memory bandwidth.


Processing-in-Memory (PIM)~\cite{DBLP:journals/corr/abs-2012-03112,
  upmem}, a.k.a.~near-data-processing, is emerging as a key technique
for reducing costly data movement.  By integrating computing resources
in memory modules, PIM enables data-intensive computation to be
executed in the PIM-enabled memory modules, rather than moving all
data to the CPU to process.  Recent studies have shown that, for
programs with high data-intensity and low cache-locality, PIM provides
significant advantages in increasing performance and reducing power
consumption by reducing data
movement~\cite{DBLP:conf/hpca/GiannoulaVPKFG021,gomez2021benchmarking}.
Although proposals for processing-in-memory/processing-in-storage date
back to at least 1970~\cite{Stone1970}, including forays by the
database community in active disks~\cite{DBLP:conf/vldb/RiedelGF98},
PIM is emerging today as a key technology thanks to advances in
3D-stacked memories~\cite{jeddeloh2012hybriddram,lee2014hbm} and the recent availability of commercial PIM
system prototypes~\cite{upmem}.
Typical applications exploiting state-of-the-art PIM architectures include
neural networks~\cite{angizi20188cmppim,li2019p3m,wang2020gnnpim,kim2021zpim},
graph processing
\confversiononly{~\cite{ahn2015tesseract, zhang2018graphp,huang2020heteropim},}
\fullversiononly{~\cite{ahn2015tesseract, zhang2018graphp,zhuo2019graphq,song2018graphr,nai2017graphpim,huang2020heteropim,newton2020pimgraphscc},}
databases~\cite{boroumand2016lazypim,boroumand2019conda},
sparse matrix multiplication~\cite{mutlu2022sparsep,xie2021spacea},
\confversiononly{and genome analysis~\cite{angizi2020pimassemble,zhang2021pimquantifier}.}
\fullversiononly{genome analysis~\cite{angizi2020pimassemble,zhang2021pimquantifier}
and security~\cite{gu2016leveraging,arafin2020pimsecurity,gupta2021invited}.}

PIM systems are typically organized as a host (multicore) CPU that
pushes compute tasks to a set of $P$ \textit{PIM modules}
(compute-enhanced memory modules), and collects the results.  Thus,
cost is incurred for moving both task descriptors and data---the sum
of these costs is the \textit{communication cost} between the CPU and
the PIM modules.  The host CPU can be any commodity multicore
processor, and is typically more powerful than the wimpy CPUs within
the PIM modules.  Thus, an interesting feature of a PIM system is the
opportunity to use both sets of resources (CPU side and PIM side) in
service of applications.

In this paper, we focus on designing a PIM-friendly ordered index for
in-memory data.  Ordered indexes (e.g., B-trees~\cite{comer1979ubiquitous}) are one of the
backbone components of databases/data stores, supporting efficient
search queries, range scans, insertions, and deletions.  Prior
works targeting PIM~\cite{liu2017concurrent,choe2019ndp} proposed ordered indexes based on
\textit{range partitioning}: the key space is partitioned into $P$
subranges of equal numbers of keys, and each of the $P$ PIM
modules stores one subrange.
Each PIM module maintains a local index over the keys in its subrange,
and the host CPU maintains the top portion of the index down to the
$P$ roots of the local indexes.  This approach works well for data and
queries with uniformly random keys---the setting studied by these
works---but it suffers from load imbalance under data or query skew.
In more realistic workloads, batches of queries/updates may
concentrate on the data in a small subset of the partitions,
overwhelming those PIM modules, while the rest are idle.  In the extreme,
only one PIM module is active processing queries and the rest are
idle, fully serializing an entire batch of queries on a single
(wimpy) processor.
The approach also suffers the cost of all data movements required to
keep partitions (roughly) balanced in size. In a recent 
paper~\cite{kang2021processing}, we designed a PIM-friendly \skiplist
that asymptotically overcomes this load-imbalance problem (details in
\Cref{sec:pim_balanced_index}), but the solution is not
practical (up to $69.7\times$ slower than the ordered index
we present in this paper).

To address the above challenges with query and data skew, we present
the \textit{PIM-tree}, a practical ordered index for PIM that
achieves both low communication cost and high load balance, regardless
of the degree of skew in data and queries.  Our skew-resistant
index is based on a novel division of labor between the host CPU and
PIM nodes, which leverages the strengths of each.  Moreover, it
combines aspects of both a \bplustree and a \skiplist to achieve its
goals.  We focus on achieving \textit{high-throughput}, processing
\textit{batches} of queries at a time in a bulk-synchronous
fashion. The PIM-tree supports a wide range of batch-parallel
operations, including point queries (\GET, \PREDECESSOR),
updates (\INSERT, \DELETE), and range \SCAN.

We introduce \textit{push-pull search}, which dynamically decides,
based on workload skew, whether (i) to \textit{push} queries from the CPU to a
PIM-tree node residing on a PIM module or (ii) to \textit{pull} the tree-node's
keys back to the CPU.  Combined with other PIM-friendly
optimizations---\textit{shadow subtrees} and \textit{chunked \skiplists}---our PIM-tree provides high-throughput, (guaranteed) low
communication costs, and (guaranteed) high load balance, for batches
of point queries, updates, and range scans.
\revision{For
example, each point query and update is
answered using only $O(\log_B \log_B{P})$ expected communication cost, 
where $B$ is the expected fanout of a PIM-tree node and $P$ is the number
of PIM modules, independent of the number of keys $n$ in the data structure, or
the data skew. Note that it would take over $10^{19}$
PIM modules for $\log_B \log_B{P}$ to exceed $1$, under our selection of $B = 16$; hence, the communication cost is constant in practice.
}

We implement the PIM-tree on the latest PIM system from
UPMEM~\cite{upmem}, with 32 CPU cores and \numDPUs PIM modules.
\fullversiononly{Our codes can be found at \githubLink.}
\revision{
We choose four \stoa ordered indexes as competitors, including
two PIM-friendly approaches~\cite{liu2017concurrent,kang2021processing}
implemented by ourselves, and two
traditional approaches\cite{DBLP:conf/ppopp/BronsonCCO10,brown2017thesis}
implemented in the SetBench benchmark suite~\cite{DBLP:conf/usenix/Arbel-Raviv0018}.
On workloads with \dataSize keys and batches of \batchSize queries,
the \pimtree{} achieves
(i) up to $59.1\times$ higher throughput than the range-partitioned
solution~\cite{liu2017concurrent},
(ii) up to $69.7\times$ higher throughput than the prior skew-resistant solution~\cite{kang2021processing},
and (iii) comparable throughput in all cases regardless of skew and as
low as $0.3\times$ less communication than two \stoa non-PIM indexes~\cite{brown2017thesis,DBLP:conf/ppopp/BronsonCCO10}.
}

The main contributions of the paper are:
\confversiononly{\vspace{-0.05in}}
\begin{itemize}[leftmargin=*]
\item
  We design the \pimtree{}, a high-throughput skew-resistant PIM data
  structure that efficiently supports a wide range of batch-parallel
  point queries, updates, and scans, even under highly-skewed workloads.
  It causes nearly constant data movement (communication cost) for a
  point query or update, and linear data movement
  for scans.  Key ideas include push-pull search and shadow subtrees.
\item
  We implement and evaluate the PIM-tree on a commercial PIM system
  prototype, demonstrating significant performance improvements at
  modest skew, and performance gains that increase linearly with larger
  skew. As far as we know these are the first implementations of an
  ordered index on a real PIM system.
\end{itemize}

\section{Background}

\subsection{PIM System Architecture \revision{and Model}}

\revision{\myparagraph{The Processing-in-Memory Model.}
We use the \emph{Processing-in- Memory Model (PIM Model)} (first described in \cite{kang2021processing})
as an abstraction of generic PIM systems.}
It is comprised of a host CPU front end
(\textit{CPU side}) and a collection of \revision{$P$} \textit{PIM modules}
(\textit{PIM side}). The CPU side is a standard multicore
processor, with an on-chip cache \revision{of $M$ words}.
Each PIM module is comprised of a DRAM memory bank (\textit{local PIM memory}) with an
on-bank processor (\textit{PIM processor}) and a local memory
\revision{of $\Theta(n/P)$ words (where $n$ denotes the problem size)}.
The PIM processor is
simple but general-purpose (e.g., a single in-order core capable of
running C code). The CPU host can send code to the PIM modules, launch
the code, and detect when the code completes. It can also
send data to and receive data from PIM memory. The model assumes there is no direct
PIM-to-PIM communication, although we could take advantage of such
communication on PIM systems supporting it.


\revision{
    As the PIM model combines a shared-memory side (CPU and its cache) and
    a distributed side (PIM modules), algorithms are analyzed using
    both shared-memory metrics (work, depth) and distributed metrics (local work, communication time).
    On the CPU side, the model accounts for \defn{CPU work} (total work summed over all cores)
    and \defn{CPU depth} (all work on the critical path).
    On the PIM side, the model accounts for
    \defn{PIM time}, which is the maximum local work on any one PIM core,
    and \defn{IO time}, which is the maximum number of messages to/from
    any one PIM module.\footnote{There is no separate accounting needed for messages to/from the CPU side
    because any well-balanced system should provide bandwidth out of the host CPU that matches bandwidth into the PIM modules (and vice-versa).}
    Programs execute in bulk-synchronous rounds~\cite{valiant1990bridging},
    and the overall complexity metrics of an algorithm is the sum of the complexity metrics of each round.
    We focus on IO time and IO rounds in this paper.
}

\myparagraph{\revision{Programming Interface.}}
For concreteness, we assume the following programming interface for our generic PIM system,
although our techniques would also work with other interfaces.
Programs consist of two parts: a \textit{host program} executed on the host CPU,
and a \textit{PIM program} executed on PIM modules.
The host program has additional functions (discussed below) to
communicate with the PIM side, including functions to invoke PIM programs on PIM modules
and to transfer data to/from PIM modules. The
PIM program is a traditional program (no additional functions) that is
invoked in all PIM processors when launched by the host program.  It
executes using the module's local memory, with no visibility into the
CPU side or other PIM modules. The specific functions are (named MPI-style~\cite{gropp1999using}):

\begin{itemize}[leftmargin=*]
    \item \textbf{PIM\_Load(}PIM\_Program\_Binary\textbf{)}: loads a binary file to the PIM modules.
    \item \textbf{PIM\_Launch()}: launches the loaded PIM
      program on all PIMs.
    \item \textbf{PIM\_Status()}: checks whether the PIM program has
      finished on all PIMs.
    \item \textbf{PIM\_Broadcast(}src, length, PIM\_Local\_Address\textbf{)}: copies a fixed length buffer to the same local memory address in each PIM module.
    \item \textbf{PIM\_Scatter(}srcs[], length[], PIM\_Local\_Address\textbf{)}: similar to \pimbroadcast, but with a distinct buffer and length for each PIM module.
    \item \textbf{PIM\_Gather(}dsts[], length[], PIM\_Local\_Address\textbf{)}: the reverse of \pimscatter, reading into the buffer array dsts[].
\end{itemize}

\confversiononly{\vspace{-0.2in}}
\begin{tboxalg}{Batch-parallel Execution($O$: batch of operations)}\label{alg:pattern}\vspace{-0.05in}
Repeat the following steps until done processing $O$:
	\begin{enumerate}[ref={\arabic*}, topsep=1pt,itemsep=0ex,partopsep=0ex,parsep=1ex, leftmargin=*]
        \item Prepare a buffer of tasks for each PIM module.
        \item Scatter the task buffers to the local memory of each PIM module using either \pimscatter or \pimbroadcast.
        \item Launch PIM programs using \pimlaunch, to run their tasks and fill their reply buffers. Wait until all the tasks finish (\pimstatus).
        \item Gather reply buffers from the PIM local memories using \pimgather.
    \end{enumerate}
\end{tboxalg}
\confversiononly{\vspace{-0.1in}}

Based on this interface, our PIM-friendly ordered index processes
\revision{batches of operations in \textit{bulk-synchronous rounds}}, like in
\cite{sewall2011palm}, using the steps in Algorithm~\ref{alg:pattern}.

As discussed in \Cref{sec:pipelining}, when implementing our
PIM-friendly programs, we use pipelining to overlap the above steps,
e.g., overlapping step~1 at the CPU and step~3 at the PIM modules.

\begin{figure}[t]
    \centering
    \includegraphics[width=0.9\linewidth]{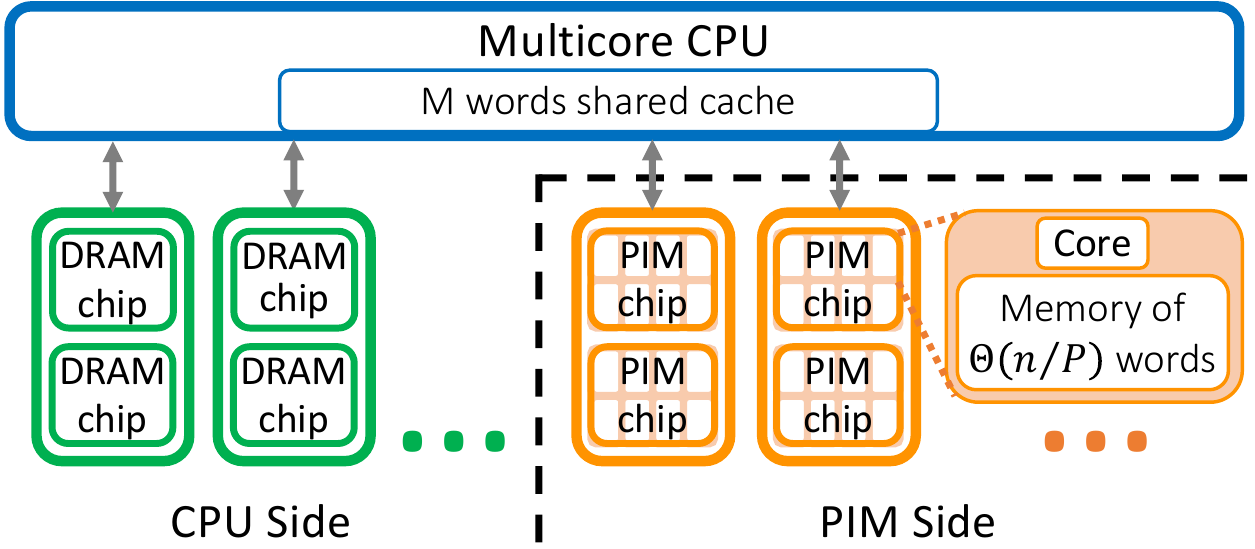}
    \vspace{0.1in}
    \caption{The architecture for the UPMEM PIM system, a specific example of
    our generic PIM system architecture.
    PIM modules are packed into memory DIMMs connected to the host CPU via normal memory channels.
    \revision{The CPU side also includes traditional DRAM modules,
    which are not part of the PIM model.}}
\label{fig:pim_architecture}
\end{figure}

\myparagraph{A Concrete Example: UPMEM.} We evaluate our techniques on
the latest PIM system from UPMEM~\cite{upmem}.  UPMEM's architecture
(\Cref{fig:pim_architecture}) is one way to instantiate
the PIM model.
Its PIM modules are plug-and-play DRAM DIMM replacements, and therefore can
be configured with various ratios of traditional DRAM
memory to PIM-equipped ones (current maximum available
configuration has 2560 PIM modules).
The CPU has access to both the
main memory (traditional DRAM) and all the PIM memory, but each PIM processor only has
access to its local memory, because PIM modules are physically
separated in different memory chips.
\revision{
Each PIM module
has up to $628$ MB/s local DRAM bandwidth, so a machine with $2560$
PIM modules can provide up to $1.6$ TB/s aggregate bandwidth~\cite{gomez2021benchmarking}.
}
To move data \textit{between}
PIM modules, the CPU reads from the origin and writes to the target.
UPMEM's SDK supports the programming interface functions listed above,
but with the restriction that the scatter/gather functions must
transmit same length buffers to/from all PIM modules (i.e., the
buffers are \textit{padded} out to equal lengths).

\revision{UPMEM's main memory (a component not in the PIM model)
enables running programs
with CPU-side memory footprints over $M$ words, but these additional memory accesses
bring another type
of communication not existing in the PIM model: \textit{CPU-DRAM communication}.
Thus it is important to write programs with good cache efficiency.
Our solution in the PIM-tree is to use only a small amount of CPU-side memory:
$\Theta(S) < M$ words for a batch of $S$ operations.}

\subsection{Load Balance Preliminaries}
\label{sec:load_balance}

A key challenge for PIM systems is to keep load balance among the PIM modules,
which we define as follows:

\begin{definition} 
  A program achieves \defn{load balance} if the \textit{work}
  (unit-time instructions) performed by each PIM program is $O(W/P)$
  and the \textit{communication} (data sent/received) by each PIM
  module is $O(C/P)$, where $W$ and $C$ are the sums of the work and
  communication, respectively, across all $P$ PIM modules.  For
  programs with multiple bulk-synchronous rounds, the program achieves
  load balance if each round achieves load balance.
\end{definition}


The challenge in achieving load balance is that the PIM module with
the maximum work or communication must be bounded.  Note that
randomization does not directly lead to load balance, e.g.,
randomly scattering $P$ tasks of equal work and communication to $P$ PIM
modules fails to achieve load balance. This is because one of the PIM
modules receives $\Theta(\log{P}/\log\log{P})$ tasks with high
probability (\whp)\footnote{We use $O(f(n))$ with high probability (\whp{}) (in $n$) to mean $O(cf(n))$ with probability at least $1-n^{-c}$ for $c \geq 1$.} in $P$~\cite{berenbrink2008weighted}, causing the work and communication at that
module to be a factor of $\Theta(\log{P}/\log\log{P})$ higher than balanced.

\revision{
We use balls-into-bins lemmas to prove load balance, where
a bin is a PIM module and a ball with weight $w$ corresponds
to a task with $w$ work or $w$ communication.
    We will use the following:


\fullversiononly{
    \begin{lemma}[\cite{sanders1996competitive}]
        \label{lemma:balls_into_bins}
        Placing $m = \Omega(P\log{P})$ balls into $P$
        bins uniformly randomly yields $O(m/P)$ balls in each bin \whp.
    \end{lemma}
}

\begin{lemma}[\cite{kang2021processing, sanders1996competitive}]
    Placing weighted balls with total weight $W = \sum{w_i}$ and
    each $w_i < W / (P \log P)$ into $P$ bins uniformly randomly
    yields $O(W/P)$ weight in each bin \whp.
\label{lemma:weightbalanced}
\end{lemma}

\fullversiononly{
\newcommand{\wbibt}{
    Placing $m = \Omega(P\log{P})$
    weighted balls with weights following a geometric distribution of expectation $\mu$
    into $P$ bins uniformly randomly yields $O(\mu m/P)$ weight in each bin \whp.
}

\begin{lemma}
    \wbibt
\label{lemma:geometric_balanced}
\end{lemma}

We prove \Cref{lemma:geometric_balanced} using the following lemma:

\begin{lemma}[\cite{brown2011wasted}]
    \label{lemma:sum_of_geometric}
    Let $Y(n, p)$ be a negative binomially distributed random
    variable that arises as the sum of n independent identically geometrically distributed
    random variables with expectation p.

    Then $E[Y(n, p)]$ = np, and for $k > 1$, $Pr[Y(n, p) > knp] \le \exp(\frac{-kn(1-1/k)^2}{2})$
\end{lemma}

\begin{proof}[Proof of \Cref{lemma:geometric_balanced}]
    Each bin gets $O(m/P)$ balls \whp, according to \Cref{lemma:balls_into_bins}.
    Then the sum of the weights of the balls in each bin is $O(\mu m/P)$ \whp
    according to \Cref{lemma:sum_of_geometric}.
\end{proof}
}
}

\subsection{Prior Work on Indexes for PIM}

There are several prior works for indexes on PIM systems. Two prior works~\cite{liu2017concurrent,choe2019ndp} proposed
PIM-friendly \skiplists. Their \skiplists are based on \textit{range partitioning}: they partition the \skiplist by disjoint key ranges and
maintain each part locally on one PIM module.
As discussed in \Cref{sec:introduction}, such range partitioning can suffer from severe load imbalance
under data and query skew.

The load imbalance problem of range-partitioned ordered indexes is also studied in traditional distributed settings.
Ziegler et al.~\cite{ziegler2019designing} discussed other choices for tree-based ordered indexes in order to avoid load imbalance, including:
(i) partitioning by the hash value of keys, (ii) fine-grained partitioning that randomly distributes all index nodes, and
(iii) a hybrid method that does fine-grained partitioning in leaves, and range partitioning for internal nodes.
They also experimentally evaluated their approaches on an 8 machine cluster.
However, each of these choices has its own problem in the case of a PIM system with thousands of PIM modules:
(i) partitioning by hash makes range operations costly, because they must be processed by
all PIM modules, (ii) fine-grained partitioning causes too much communication because all accesses will be non-local, and
(iii) the hybrid method suffers from the load balance problem in its range partitioned part.

\begin{figure}[t]
    \centering
    \includegraphics[width=\linewidth]{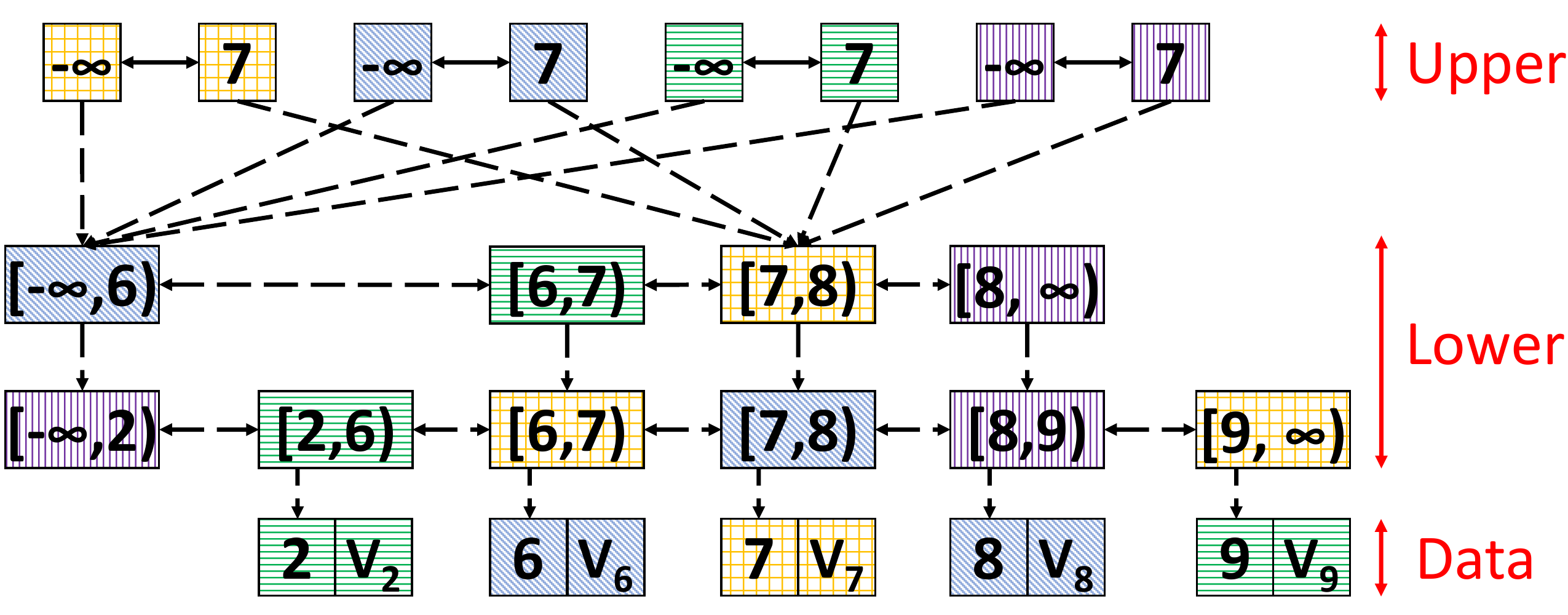}
    \vspace{0.1in}
    \caption{PIM-balanced \skiplist~\cite{kang2021processing}
 with the upper part replicated on a 4-PIM system. Nodes on different
    PIM modules are different colors. PIM pointers are dashed lines.
    The lower part is $\log{P}$ levels.
    }
\label{fig:replicated_upper_part}
\end{figure}

\subsection{Prior Work: PIM-balanced \SkipList}
\label{sec:pim_balanced_index}

\revision{In a recent paper~\cite{kang2021processing}, we presented
the first provably load-balanced batch-parallel \skiplist index,
the \defn{PIM-balanced \skiplist}, under adversary-controlled
workloads on the PIM model.}  A key insight was to
leverage the CPU side to solve the load balance problem.


The \textbf{PIM-balanced \skiplist} horizontally splits the \skiplist
into two parts, an \textit{upper part} and a \textit{lower part},
replicating the upper part in all PIM modules and distributing lower
part nodes randomly to PIMs.  This is shown in
\Cref{fig:replicated_upper_part}, where nodes in different PIM modules
have different colors, and the replicated upper part is explicitly drawn as four copies.
For a system with $P$ PIM-modules, the lower part is $\log P$ levels.
We can afford to replicate (only) the top part because
(i) it is small relative to the rest of the \skiplist and (ii) it is updated relatively
infrequently (recall that an inserted key reaches a height $h$ in a \skiplist with probability
$1/2^h$).

Queries are executed by pointer chasing in the ``tree'' of \skiplist
nodes. The batched queries are first evenly divided and sent to all PIM modules, each progressing through the
upper part locally.
Then the \skiplist goes through the lower part by sending the
query to the host PIM module of each lower part node on the search
path one-by-one, until reaching a leaf. We call
this the \defn{Push} method, because queries are sent (``pushed'') to
PIM modules to execute.

Executing a batch of parallel queries using only \textit{Push} can
cause severe imbalance, despite the lower part nodes being randomly
distributed.  For skewed workloads, many queries may share a common
node on their search path, meaning that they are all sent to the host
PIM module of that node, causing a load imbalance. These nodes are called
\textit{contention points}.  An example is when multiple non-duplicate
\PREDECESSOR queries
return the same key, with \textit{all} nodes on the search path to
that key being contention points.

The \textbf{PIM-balanced \skiplist}~\cite{kang2021processing} solves this problem by avoiding
contention points, based on a key observation: once the search
paths of keys $l$ and $r$ share a lower part node $v$, searching any key
$u \in [l, r]$ will also reach node $v$. Thus the search for $u$ can start directly from
the LCA (lowest common ancestor) of these two paths. We call this the
\defn{Jump-Push} method.  \textit{Jump-Push} search has a
preprocessing stage to record search paths. It is a multi-round sample
search starting with one sample: In each round, it doubles the sample
size and uses the search paths recorded in previous rounds to decide
start nodes of sample queries in this round.  This approach
limits the contention on each node, avoiding load imbalance.

However, the preprocessing cost is high.  For $P$ PIM modules and a
batch of $P \log^2 P$ operations, it takes $O(\log P)$ sampling rounds,
each of which takes $O(\log P)$ steps of inter-module pointer chasing
to search the lower part.  The main stage, in contrast, takes only
$O(\log P)$ steps. Moreover, the preprocessing stage requires
recording entire search paths---another overhead for the CPU
side.

Our new ordered index (PIM-tree) uses some of the same ideas as
this work, but includes key new ideas to make it simpler, and more
efficient both theoretically and practically.

\section{\pimtree{} Design}
\label{sec:method}

\myparagraph{Overview.}
The \pimtree{} is a batch-parallel skew-resistant ordered index designed for PIM systems.
It supports fundamental key-value operations, including \GET{(key)},
\UPDATE{(key, value)}, \PREDECESSOR{(key)}, \INSERT{(key, value)}, \DELETE{(key)},
and \SCAN{(Lkey, Rkey)}. It executes operations in same-type atomic
batches in parallel, similar to \cite{sewall2011palm}.
As such, the data structure avoids conflicts caused by operations of
different types. We design it starting from the structure discussed in \S\ref{sec:pim_balanced_index}.

In this section, we describe \pimtree{}'s design by studying the
impact of our techniques/optimizations on the \PREDECESSOR{} operation.
We review the basic data structure design discussed in
\S\ref{sec:pim_balanced_index} in detail, then introduce our three key
techniques/optimizations, and finally analyze the resulting
communication cost and load balance. Later in \S\ref{sec:pim_tree}, we
describe the design of \pimtree{}'s other operations.

\myparagraph{Notations.}
We refer to the {number of PIM modules} as ${P}$, the {total
number of elements in the index} as ${n}$,
the {batch size} as ${\batch}$, and {the expected fanout of PIM-tree nodes}
as ${B}$.

\myparagraph{Key Ideas.}
\label{sec:key_ideas}
We observe that the two components of the PIM architecture, the CPU side and the PIM side,
prefer different workloads.
The distributed PIM side prefers uniformly random workloads and
suffers from the load imbalance caused by highly skewed ones.
Meanwhile, the shared CPU side prefers skewed workloads since we can explore spatial and temporal locality in these workloads that leads to better cache efficiency.
This difference shows the complementary nature of shared-memory and
distributed computing, and their coexistence in the PIM architecture
motivates our hybrid method:
we design a dynamic labor division strategy that automatically
switches between the CPU side and the PIM side to use the more ideal
platform.
This strategy, called \textit{\PushPullSearch}, is the core technique of the PIM-tree.

With load balance achieved by \PushPullSearch, we further propose two other optimizations, called \textit{shadow subtrees} and \textit{chunking},
to reduce communication.
These optimizations are
motivated by two basic ideas respectively: caching remote accesses at PIM modules to
build local shortcuts (thereby eliminating communication), and blocking nodes into chunks (for better locality).  We will show
how these ``traditional'' techniques combine
with the \PushPullSearch optimization to bring an asymptotic reduction
of communication from $O(\log{P})$ to $O(\log_B\log_B{P})$ ($B$ is the expected fanout of chunked nodes),
and a throughput increase of up to $69.7\times$, compared with the
PIM-friendly skip list of \S\ref{sec:pim_balanced_index}.

\subsection{Basic Structure}\label{sec:basic_structure}

The \textit{PIM-balanced \skiplist} is a distributed \skiplist
horizontally divided into three parts, the \textit{upper part}, the
\textit{lower part}, and the \textit{data nodes}
(\Cref{fig:replicated_upper_part}).  Data nodes are key-value pairs
randomly distributed to PIM modules to support hash-based lookup in
one round and $O(1)$ communication. Every upper part node is
replicated across all PIM modules, and every lower part node is stored
in a single random PIM module.  The ID of the PIM module hosting a
lower part node is called the node's \defn{PIM ID}.  Remote pointers,
called \defn{PIM pointers}, are comprised of (PIM ID, address) pairs.
In \Cref{fig:replicated_upper_part}, PIM pointers are represented by
dashed arrows, while traditional (i.e., intra-PIM) pointers are
represented by solid arrows.  To save communication during search,
each lower part node stores the key of the next \skiplist node at the
same level (called a \defn{right-key}).  There is a node at the lowest level for each key, and the probability of a node
joining the next higher level in the \skiplist is set to $1/2$.

The upper part is replicated to enable local executions for queries on
PIM modules, but the replication brings an overhead of $P$ to both
space complexity and update costs. To mitigate this overhead,
the lower part height is set to be $H_{\text{low}} = \log P$, so that only a
${1/2}^{\log P} = 1/P$
fraction of the keys reach the upper part. By replicating the upper
part, the number of remote accesses needed for a \PREDECESSOR query is
reduced from $O(\log n)$ to $O(\log P)$.

In the following sections, we call the upper part \defn{L3} and the lower
part \defn{L2}.  After applying the Shadow Subtree optimization
(\S\ref{sec:shadow_subtree}), we will further divide the lower part
horizontally into two parts, called \defn{L2} and \defn{L1}.




\subsection{\PushPull{} Search}
\label{sec:push_pull_search}



\PushPull{} search is our proposed search method that guarantees load
balance even under skewed workloads.
In the \defn{\Push} method, the CPU sends the query to
the host PIM module of the next node along the search path, the PIM
module runs the query, then the CPU fetches the result; in the
\defn{\Pull} method, the CPU retrieves the next node along the path
back to the CPU side, running the query itself.
\defn{\PushPull{}}
search chooses between \Push and \Pull by counting the number of
queries to each node: when the number of queries to a node exceeds a
specific threshold, denoted as ${K}$ in the following
sections, we \Pull that node, otherwise we \Push the query.


In further detail, \PushPull{} search performs multi-round pointer
chasing over the basic structure mentioned in
\S\ref{sec:basic_structure} in three stages, where the CPU records the next
pointer for each query as an array of PIM pointers throughout the process.
\begin{enumerate}[label=(\arabic*),topsep=0pt,itemsep=0pt,parsep=0pt,leftmargin=15pt]
  \item \underline{Traverse L3 using the replicated upper parts.} The CPU
  evenly distributes queries to PIM modules.  Each PIM module runs its
  queries using its local copy of L3, until reaching a pointer to an L2
  node. The CPU retrieves these pointers (using \pimgather).

  \item \underline{Traverse L2 using contention-aware \PushPull.} The CPU
  performs multiple \PushPull{} rounds.  In each round, the CPU counts
  the number of queries to each L2 node. If there are more than $K$
  queries to a node, choose \Pull by sending a task to the PIM-side to
  retrieve that node to the CPU-side, then partition the queries (in parallel) based
  on the PIM IDs in the retrieved node's pointers on the CPU-side.
  Otherwise, choose \Push to send a Query task to the PIM and
  retrieve the next pointer for the query.
  \item When the search reaches a data node, return the data.
\end{enumerate}

We can record the addresses of all nodes on the pointer-chasing path
for a query on the CPU side to get the \defn{search trace} for each
query. Note that these traces are used when performing updates (in
\S\ref{sec:pim_tree}).
For the basic structure mentioned in \S\ref{sec:basic_structure}, we
choose $K = 1$, as it minimizes communication for constant size
nodes.

\revision{
\myparagraph{Discussion.}
The most interesting part of \PushPull search is that it is based on
integrating two fundamental methods from distributed and shared-memory
computing to achieve provable load balance with low cost
(see \S\ref{sec:proof_for_tree_search} for analysis). We
observe that the \Push method is a distributed computing technique, as
it uses the CPU as a router and always runs queries on PIM
modules. Meanwhile, the \Pull method is a shared memory technique,
treating the PIM modules as standard memory modules and
running the queries on the CPU.
As discussed in \S\ref{sec:key_ideas}, combining such fundamental methods
works because of the complementary nature of
the CPU side and the PIM side in the load balance issue:
contention-causing (thus PIM-unfriendly) workloads
are meanwhile CPU-friendly workloads.

As a solution only to the load balance issue, \PushPullSearch
provides no asymptotic improvements in worst-case bounds compared
with \PushOnly or \PullOnly methods.
Such improvements are
provided by our optimizations, \textit{shadow subtrees} and
\textit{chunking}, which we describe next.

}

\vspace{-0.15cm}
\subsection{Shadow Subtrees}\label{sec:shadow_subtree}


Shadow subtrees are auxiliary data structures in L2 that act as
shortcuts to reduce communication from $O(\log{P})$ to
$O(\log\log{P})$ for each query\revision{, while ensuring that the
space complexity is still $O(n)$}.
The shadow subtree optimization is
based on the idea of the search tree defined by a \skiplist, which is
an imaginary tree generated by merging all possible search paths of a
\skiplist. It contains all nodes and all edges of the \skiplist,
except some horizontal edges. The \defn{shadow subtree} of each node
is a shadow copy of its search subtree stored together with this
node. By using shadow subtrees, a PIM module can run queries locally
through L2.  Although shadow subtrees and replicating the top of the tree
both involve copying nodes across different PIM modules with the
purpose of reducing communication, they are actually quite different.
When replicating the top, a single tree is copied $P$ times across
the modules.  In the shadow subtree, every ancestor of a node has a
copy of that node as part of its shadow subtree (in our case just the
ancestors in L2).

\revision{
Building shadow subtrees
    on all ($O(n)$) L2 nodes would require $O(n\log{P})$ space.
    Instead, to maintain $O(n)$ space, we build them only on a small proportion of L2 nodes.
In particular, we divide L2 into two layers, denoting the upper
levels to be the new L2 and the lower levels to be L1. We build shadow
subtrees only on the new L2.
We set the height of L1 to be $H_{\text{L1}} = \log\log{P}$, so
 only $(1/\log{P})$-fraction of nodes ($O(n/\log{P})$ nodes) are in the new L2,
and the space complexity summing over all shadow subtrees is $O(n)$.
Thus, the PIM-tree now has three layers: L3 under full replication, L2 with shadow subtrees,
and L1 under random distribution without any replication. Each layer
requires $O(n)$ space, so the total space complexity is $O(n)$.
}
This is shown in
\Cref{fig:shadow_subtree}. We refer to original tree nodes and pointers
to them as \defn{physical} nodes (pointers), and mark them in black.
Shadow-tree nodes and pointers to them are referred to as
\defn{shadow} nodes and pointers, and are marked in red.

\begin{figure}[t]
    \centering
    \includegraphics[width=0.9\linewidth]{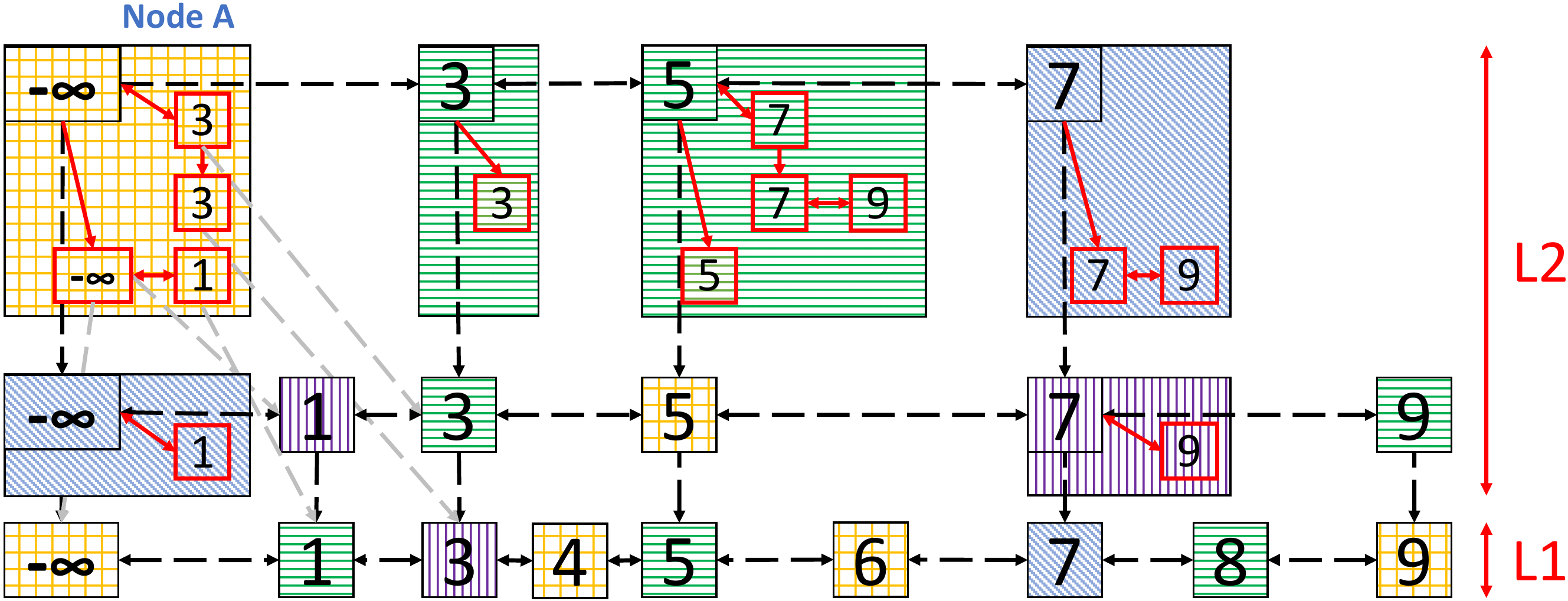}
    \caption{The structure of L2 and L1 after introducing shadow subtrees.
    Shadow nodes and shadow pointers are marked in red.
    Note that blue $1$ does not have a shadow tree node for $3$ because node $3$ is not in its search subtree.
    Right-keys are omitted.
    We also omit pointers from shadow nodes to physical nodes except for node A. The L1 part (L2 part) is $\log\log P$ levels ($\log P - \log\log P$ levels, respectively).
    }
\label{fig:shadow_subtree}

    \centering
    \includegraphics[width=0.9\linewidth]{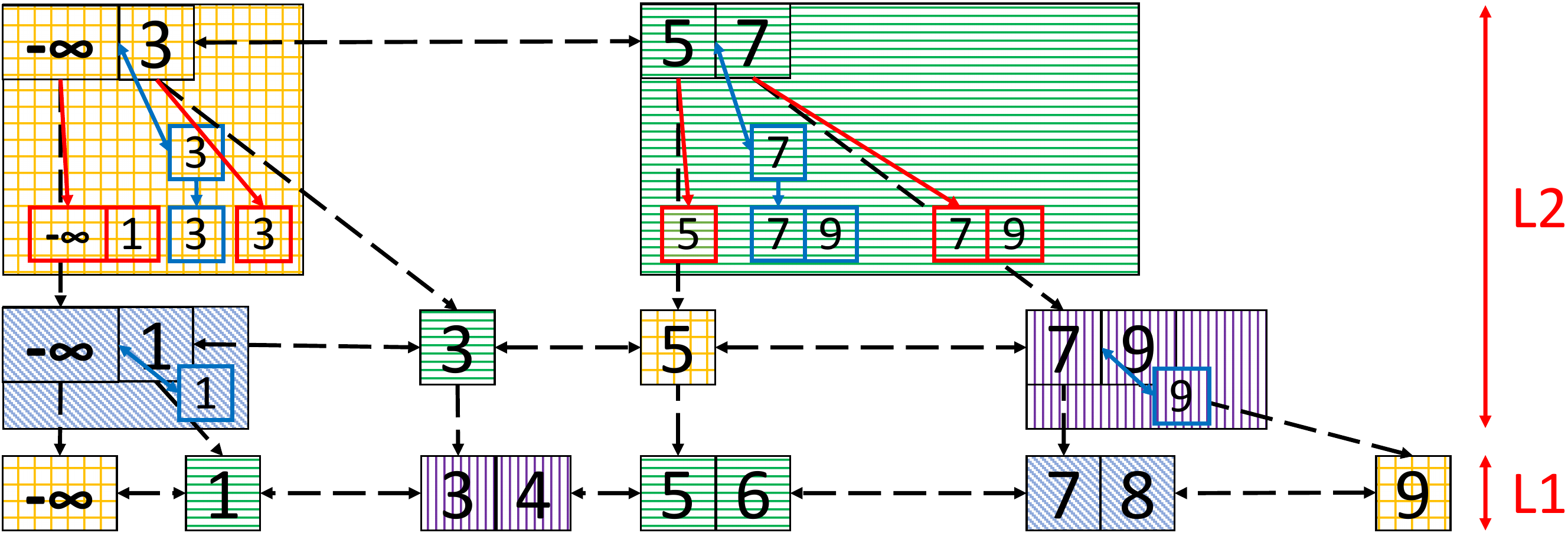}
    \caption{\revision{The intermediate state of the chunking transformation.
    We merge non-pivot
    nodes (nodes whose keys do not go to upper levels) to their
    left-side neighbors.
    Redundant shadow subtrees after merging are marked in blue, and will be removed
    in \Cref{fig:pim_tree}.
    All physical pointers from shadow nodes are omitted.}
    }
\label{fig:shadow_subtree_transformation}

    \centering
  \vspace{0.1em}
    \includegraphics[width=0.9\linewidth]{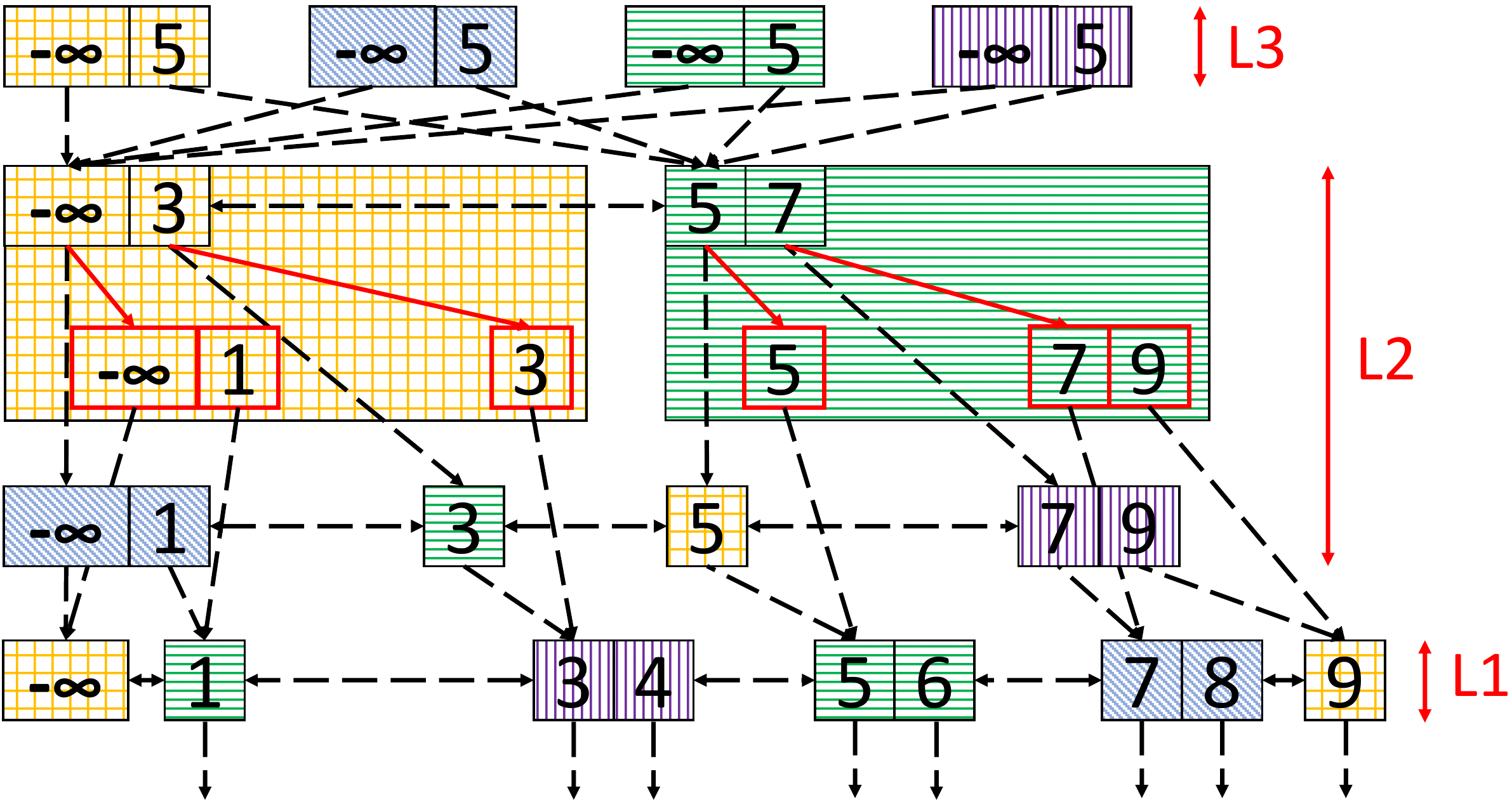}
    \caption{The actual structure of the \pimtree{} \revision{with redundant shadow
    subtrees removed from \Cref{fig:shadow_subtree_transformation}}.
    We no longer need right-keys after chunking. Data nodes are omitted.}
\label{fig:pim_tree}
   \vspace{-0.15in}
\end{figure}

\vspace{-0.05cm}
\myparagraph{Accelerating \PREDECESSOR{} using Shadow Subtrees.}
Shadow subtrees strengthen the \Push side of \PushPull{}
search: a single \Push round can send a query through the whole L2,
rather than going forward by just one level, by running the query on
the shadow subtree.
Therefore the search process takes only $O(\log\log{P})$ rounds for
uniform-random workloads: one \Push through L3, one \Push through L2,
and $O(H_{\text{L1}}) = O(\log\log{P})$ \PushPull{} rounds for L1.
However, for skewed workloads, we cannot simply perform a single \Push
round through L2, because multiple queries may still be pushed to the
contention points in L2 and cause load imbalance.  We solve this problem
again by \Pull, by introducing a multi-round \Pull process to
eliminate contention points.

In further detail, \PushPull{} search in L2 has two stages: we first
perform up to $O(H_{\text{L2}})$
\Pull rounds for nodes with $\ge K$ queries
until no such node exists, where $H_{\text{L2}} = \log{P} - \log \log{P}$ denotes the new L2's height,
then execute one ``Push'' round to send
all queries through L2.
We set the threshold $K = H_{\text{L2}}$ instead of $1$ since ``Push'' is
now more powerful and we tend to use it more.
Both stages take $O(1)$ balanced communication per query\conffulldifferent{.
\revision{This is partly proved in Lemma~\ref{lemma:push_only_proof},
and fully proved in the full paper~\cite{pimtree2022full}.}
}{, as we will prove in Lemmas~\ref{lemma:pull_only_proof}
and~\ref{lemma:push_only_proof}.
}

In practice, we use another optimization to reduce the number of
\Pull rounds.  Note that although contention points are the only source
of load imbalance, we may reach a reasonable level of load balance
before eliminating all contention points.
Therefore, to avoid unnecessary \Pull{} rounds, before starting a
\Pull round, we measure the load balance across PIMs by counting the
number of queries that will be sent to each PIM module; if the one
with the most is below $3\times$ the average load, we stop the \Pull
round and start to \Push.

\myparagraph{Replication and Space/Imbalance Tradeoffs.}
Compared with full replication used for L3 and range partitioning
(which performs no replication) used in related works,
shadow subtree is a novel scheme 
that supports queries with
$O(1)$ communication by improving the locality of a distributed ordered
index.
Specifically, shadow subtree is a selective
replication approach that lies between these two prior schemes.  If we
replicate nodes not only to their L2 ancestors, but to all PIM
modules, we obtain full replication. On the other hand, if we only
keep the shadow subtrees of the L2 roots, we obtain the range
partitioned scheme.

The cost and the skew-resistance of shadow subtrees also lies between
those of the other two schemes. \Cref{table:replications} shows
the bounds when applying different schemes to a \skiplist with height
$\log{P}$ and size $P$ in expectation.
In the full replication scheme, we can run queries with perfect load
balance, but it brings an overhead factor $P$ to both space complexity
and update costs. On the other hand, for the range partitioned scheme,
each query can only be executed by a single PIM module.  We can still
do \PushPull{} to avoid contention with a threshold of $K = P$: we
will choose to simply \Pull{} the whole tree when the number of
queries exceeds the size of the whole part. There can be load
imbalance, as some part gets up to $P$ queries and others get none. For
this approach there is no overhead on space complexity or for updates.
Lastly, using shadow subtrees, the overhead factor is
$O(\log{P})$ as each node is replicated in all its L2 search tree
ancestors, and the maximum query number is $\log{P}$ according to our
choice of \PushPull{} threshold $K = H_{\text{L2}}$.

Shadow subtrees therefore yield a balanced compromise between the
two schemes, providing a sweet spot for both overhead and
skew-resistance.



\begin{table}
    \begin{tabular}{lccc}
        \toprule
        Scheme    & Overhead factor & Maximum query number\\
        \midrule
        Full Replication & $P$   &  Perfect Balance\\
        Range Partitioned   & 1 & $P$ \\
        Shadow Subtrees   & $O(\log{P})$ & $\log{P}$ \\
    \bottomrule
    \end{tabular}
    \caption{Comparison between three types of replication schemes \revision{that run queries with $O(1)$ communication}. The larger the overhead factor, the more space it takes and the
    slower updates will be. The larger the maximum query number is, the more imbalanced the execution will be under skewed workloads.}
    \label{table:replications}
    \vspace{-0.05in}
\end{table}

\vspace{-0.05cm}
\subsection{Chunked \SkipList}\label{sec:chunked}



Chunking or ``blocking'' is a classic idea widely used in
locality-aware data structures, e.g., B-trees and \bplustrees.
To improve locality, we apply a similar chunking approach to improve
the access granularity of the PIM computation,
while decreasing the tree height.
As chunking increases the access granulariy, each PIM processor
obtains larger local memory bandwidth, therefore better performance.
The effect of access granularity in PIM
is discussed in detail in \cite{gomez2021benchmarking}.

We apply chunking to all layers of the \pimtree{}.  In L3, we replace
the multi-thread \skiplist with a batch-parallel multi-threaded \bplustree~\cite{sewall2011palm}.
\revision{In L2 and L1, we chunk the nodes in our \skiplist to obtain a \defn{chunked \skiplist}.
We first merge horizontal non-pivot nodes (whose keys do not go to
upper levels) into a single chunk, then remove redundant shadow subtrees.
Applying this two step process on \Cref{fig:shadow_subtree} first gives
\Cref{fig:shadow_subtree_transformation} as an intermediate state,
and finally the PIM-tree in \Cref{fig:pim_tree}.

The result with shadow subtrees looks similar to a \bplustree.
The difference is that while the \bplustree sends nodes to upper
levels on overflow of lower level nodes, the chunked
\skiplist uses random heights generated during \INSERT, so
the fanout holds in expectation.}
We decrease the probability of reaching the next level in the \skiplist
from $1/2$ to $1/B$, so that the expected fanout is $B$.  We
choose the same chunking factor $B$ in L3, L2 and L1 for simplicity,
but different factors could be used in each part.
\hongbo{grammar?}\laxman{looks good to me.}
As discussed in \S\ref{sec:insert},
we use a chunked
\skiplist instead of a classical \bplustree in L2 to make batch-parallel
distributed \INSERT and \DELETE simpler and more efficient.
We use a \bplustree in L3 because the structure is not distributed,
making batch-parallel \INSERT and \DELETE easier.




Chunking reduces tree height at all levels, which improves
multiple aspects of our design. We denote the new L2 (L1) height as
$H_{\text{L2}}'$ ($H_{\text{L1}}'$, respectively).
The L2 part of the search path to
each node is reduced from $O(H_{\text{L2}}) = O(\log{P})$ to
$H_{\text{L2}}'=\log_B{P} - \log_B\log_B{P}$, as there are no longer horizontal pointer-chasing processes.
Therefore, the space and replication overhead of shadow subtree from
$O(H_{\text{L2}})$ to $H_{\text{L2}}'$, as each node is only replicated in its L2
ancestors.
Furthermore, lower overhead enables us to reduce $H_{\text{L1}}'$ to
$\log_B \log_{B}{P}$.
\revision{
$H_{\text{L1}}'$ is effectively $1$ in practice,
because with our choice of $B=16$ it will
take over $10^{19}$ PIM modules for $H_{\text{L1}}'$  to exceed $1$.
}

Therefore, in practice, the height of L2 is reduced to $2$ levels, and
L1 reduced to $1$ level.  The probability for a key to reach L3 is $1/4096 <
1/P$, and the probability of reaching L2 is $1/16 < \log_{16}{P}$.


\myparagraph{Implementing \PREDECESSOR{} after Chunking.}
Chunking brings only one modification to the search process: changing
the \PushPull{} threshold $K$ from $H_{\text{L2}}$ to $B\cdot
H_{\text{L2}}'$, because we now ``Pull'' chunks with expected size
$O(B)$ instead of $O(1)$.
The detailed algorithm is explained in \S\ref{sec:proof_for_tree_search}.

Chunking also improves the communication costs of \PREDECESSOR. First,
as the height of L1 is reduced from $\log\log{P}$ to $\log_B \log_B{P}$,
each query now causes only $O(\log_B \log_B{P})$ communication in L1.
Second, chunking reduces the maximum possible number of Pull-only
rounds in L2 from $O(H_{\text{L2}}) = O(\log{P})$ to exactly
$H_{\text{L2}}'$, which is $\log_B{P} - \log_B \log_B{P}$.
This helps reduce the number of communication rounds under skewed workloads.

\vspace{-0.05cm}
\subsection{\PREDECESSOR Algorithm and Bounds}
\label{sec:proof_for_tree_search}

Next, we describe the complete algorithm for \PREDECESSOR{}, and
discuss its cost complexity.
We provide proofs for the communication cost and load balance of
\PREDECESSOR queries.
\revision{
For simplicity throughout the paper, our cost analyses assume that hash
functions provide uniform random maps to PIM modules, so that the lemmas
in \S\ref{sec:load_balance} can be applied.
}
Algorithm~\ref{alg:pred} summarizes the search process.
\confversiononly{\vspace{-0.25in}}

\begin{tboxalg}{\PREDECESSOR($Q$: batch of query keys)}\label{alg:pred}
	\begin{enumerate}[ref={\arabic*}, topsep=0pt,itemsep=0ex,partopsep=0ex,parsep=1ex, leftmargin=*]

    \item \Push queries from $Q$ evenly to PIM modules, and traverse
    L3.\label{pred:one}

    \item While the number of queries that will be sent to each PIM
    module for L2 is not balanced (i.e., the busiest PIM module gets
    more than $3 \times$ the average load), do the
    following:\label{pred:pull}
    \begin{enumerate}[ref={\arabic*}, topsep=0pt,itemsep=0ex,partopsep=0ex,parsep=1ex, leftmargin=*]
      \item \Pull all nodes with more than $K = B\cdot H_{\text{L2}}'$
        queries back to the CPU.
      \item Use these nodes to progress the
      pointer-chasing process of these queries by one step.

    \end{enumerate}

    \item \Push each query to the PIM module holding its search node, and traverse
    L2 using the shadow subtrees.\label{pred:shadow}

    \item Perform $H_{\text{L1}}'$ \PushPull rounds with $K = B$ to traverse L1, and
    retrieve the data nodes. \label{pred:pushpull}

	\end{enumerate}
\end{tboxalg}

\revision{
    \hongbo{update}
    We demonstrate here a mini step-by-step example of a \PREDECESSOR batch with
    four queries on the \pimtree{} in Figure \ref{fig:pim_tree} (note that
real batches
    should have more queries on this tree to achieve load-balance).
    The queries request the \PREDECESSOR{}s of keys 1, 3, 4 and 7.
    \pimtree{} first evenly distributes one query for each of the four PIM
    modules to search through L3, returning three queries falling onto
    the L2 node $[-\infty,3]$ and one falling onto node $[5,7]$.
    The context of node $[-\infty,3]$ will be pulled to the CPU from the
    yellow-masked PIM module due to its large contention,
    and the pointer-chasing searching of keys 1, 3 and 4 over L2 will be
    executed on the CPU side.
    After that, query 1, query (3, 4), and query 7 will be pushed to the PIM
    module containing
    the blue-masked node $[-\infty, 1]$, green-masked node $[3]$ and green-masked
    node $[5, 7]$ respectively
    on the local shadow subtrees to search through L2.
    Finally, all queries will be carried out in a similar \PushPull way
    to return the results from L1 and data nodes.
}

\confversiononly{\vspace{-0.05in}}
\begin{theorem}\label{thm:pred}
    A batch of \PREDECESSOR queries can be executed
    in $O(\log_B{P})$ communication
    rounds, with a cost of $O(\log_B\log_B{P})$ communication for each
    operation in total \whp.
    The execution is load balanced if the batch
    size $\batch = \Omega(P\log{P}\cdot B\cdot H_{\text{L2}}') =
    \Omega(P\log{P}\cdot B\cdot \log_B{P})$.
    \revision{The CPU-side memory footprint is $O(S)$.}
\end{theorem}
\confversiononly{\vspace{-0.05in}}

\conffulldifferent{
\revision{
We provide part of the proof here, and give the full proof with all
details in the full paper~\cite{pimtree2022full}.
The key challenge is to prove the communication bounds and load balance,
and we do this by proving separately for each stage of Algorithm~\ref{alg:pred}.
We take Lemma~\ref{lemma:push_only_proof}, the proof for the L2 \Push stage
(stage 3) as an example.
\vspace{-0.05in}

\begin{lemma}[\Push round for L2]\label{lemma:push_only_proof}
    \Push using the shadow subtrees (stage~\ref{pred:shadow})
    takes $1$ round, $O(1)$ communication \whp for each query, and is load balanced.
\end{lemma}
\vspace{-0.15in}

\begin{proof}
    In this stage we send each query as a task to the corresponding PIM
    module, incurring $O(1)$ communication per query and 1 round
    overall.
    For load balance, we analyze it as a weighted balls-into-bins
    game, where we take the target nodes as balls, the numbers of
    queries on the target nodes as weights, and PIMs as bins. The weight
    limit is $K = B\cdot H_{\text{L2}}'$ by assumption,
    as each node gets at most $K$ queries,
    and the weight sum is at most $\batch$. Applying
    \Cref{lemma:weightbalanced}, each PIM module incurs $O(\batch/P)$
    communication.
\end{proof}
\vspace{-0.05in}
}
}{
The following lemma is useful for bounding the communication of
queries that do not incur contention. We use \Cref{lemma:push,lemma:pull}
to prove \Cref{lemma:pull_only_proof,lemma:push_only_proof,lemma:push_pull_proof},
then combine them to prove \Cref{thm:pred}.

\begin{lemma}[Push]\label{lemma:push}
    For batch size $\batch = \Omega(K\cdot P\log{P})$,
    using \Push for queries whose target nodes have fewer than $K$
    queries incurs $O(\batch/P)$ communication per PIM module \whp.
\end{lemma}

\begin{proof}
    We analyze the communication using a weighted balls-into-bins
    game, where we take the target nodes as balls, the numbers of
    queries on the target nodes as weights, and PIMs as bins. The weight
    limit is $K$ by assumption, as each node gets at most $K$ queries,
    and the weight sum is at most $\batch$. Applying
    \Cref{lemma:weightbalanced}, each PIM module incurs $O(\batch/P)$
    communication.
\end{proof}

\begin{lemma}[Pull]\label{lemma:pull}
    For batch size $\batch = \Omega(K\cdot P\log{P})$, and nodes with
    geometrically-distributed chunk sizes, with an expected chunk size
    of $B$, using \Pull to fetch nodes with more than $K$ queries
    incurs $O(B\cdot \frac{\batch}{KP})$ communication for each PIM module \whp.
\end{lemma}

\begin{proof}
    We \Pull no more than $\batch/K$ nodes in each round. The amount of
    communication caused by using \Pull to fetch each node is equal to
    its node size, which is geometrically distributed, and is $B$ in
    expectation. Using weighted balls-into-bins again, we treat the
    balls as target nodes, communication on each target node as
    weights, and PIMs as bins.  As $\batch = \Omega(K\cdot P\log{P})$,
    we have that the number of balls is $\Omega(P\log{P})$. Thus,
    applying \Cref{lemma:geometric_balanced} with $\mu = B$, each PIM
    module incurs $O(B \cdot \frac{\batch}{KP})$ communication \whp.
\end{proof}

\begin{lemma}[\PullOnly rounds for L2]\label{lemma:pull_only_proof}
    The multi-round \Pull stages for L2 (the loop in
    Algorithm~\ref{alg:pred}) are load balanced. Overall,
    the loop incurs $O(1)$ communication and PIM work \whp per
    operation, and finishes in $H_{\text{L2}}' = \log_B{P} - \log_B\log_B{P}$
    rounds.
\end{lemma}

\begin{proof}
    Because the number of rounds is upper bounded by the L2 height, the loop
    must finish in $H_{\text{L2}}'$ rounds.
    To prove load balance and bound communication and work, we show that the CPU-PIM
    communication per PIM is $O(\frac{\batch}{H_{\text{L2}}' \cdot P})$ in
    each of the rounds.
    This result follows by applying \Cref{lemma:pull} with $K = B\cdot H_{\text{L2}}'$.
\end{proof}

\begin{lemma}[\Push round for L2]\label{lemma:push_only_proof}
    \Push using the shadow subtrees (stage~\ref{pred:shadow})
    takes $1$ round, $O(1)$ communication \whp for each query, and is load balanced.
\end{lemma}

\begin{proof}
    In this stage we send each query as a task to the corresponding PIM
    module, incurring $O(1)$ communication per query and 1 round
    overall.  The load balance follows by applying \Cref{lemma:pull}
    with $K = B\cdot H_{\text{L2}}'$.
\end{proof}

\begin{lemma}[\PushPull rounds]\label{lemma:push_pull_proof}
    For chunk nodes with geometric-distributed sizes of expectation $B$,
    the \PushPullSearch rounds (stage~\ref{pred:pushpull}) with batch
    size $\batch = B\cdot P\log{P}$ and threshold $K = B$ incur $O(H_{\text{L1}}')$
    communication and PIM work \whp for each query, and achieve load
    balance. The data nodes can be treated as chunk nodes with $B = 1$.
\end{lemma}

\begin{proof}
    As \PushPull is
    a combination of both \Push and \Pull, this can be proved by apply $K = B$ to
    both \Cref{lemma:push} and \Cref{lemma:pull}, for each of the $H_{\text{L1}}'$ rounds.
\end{proof}

\begin{proof}[Proof for Theorem~\ref{thm:pred}]
    \revision{The CPU-side memory footprint is $O(S)$ because the only
    datum needed for all steps is an array of $S$ (key, address) pairs,
    taking $O(S)$ space, and
    auxiliary memory used in each step (e.g. reading/writing the task buffer
    mentioned in Algorithm~\ref{alg:pattern}) is also $O(S)$.}

    We prove the communication bounds and load balance separately for
    each stage of Algorithm~\ref{alg:pred}:

    It's obvious that stage~\ref{pred:one} incurs $O(1)$ communication, one
    communication round. It's load balanced
    because L3 is replicated across the PIMs, and we evenly distribute the
    queries to the PIMs.

    Bounds for the Pull-only rounds and
    Push-only rounds in L2, and lastly the \PushPull{} rounds for L1
    and data nodes are proven seperately in Lemmas~\ref{lemma:pull_only_proof},
    \ref{lemma:push_only_proof} and~\ref{lemma:push_pull_proof}, respectively.
\end{proof}
}

\section{PIM-tree: Other Operations}
\label{sec:pim_tree}

Having described the design of the PIM-tree data structure in
\S\ref{sec:method}, using the \PREDECESSOR operation as the running
example, \conffulldifferent{
  we now briefly introduce how other PIM-tree operations are implemented. Please
  refer to the full paper~\cite{pimtree2022full} for more detail.
}{
  we now describe how other PIM-tree operations are implemented.
}

\subsection{\GET and \UPDATE using Hashing}\label{sec:get_and_update}

\GET and \UPDATE are operations with a given key. These operations also
do not modify the structure of the data structure. Therefore, we solve
these in one round and $O(1)$ communication per operation through a
hash-based approach by first (i) using a fixed hash function to map
keys to PIM modules, and (ii) building a local hash table on each PIM
module to map keys to the local memory addresses of their data nodes.

Because the data nodes are distributed by a hash function, we achieve
good load balance even for skewed workloads, assuming that there are
not duplicate operations to the same key. If such redundant operations
exist, we can solve this by preprocessing to combine operations on
the CPU-side, using a user-defined combinining mechanism.
In practice, we use a linear-probing hash table on the PIM-side, but
other hash-table variants could also be used.

\fullversiononly{
Algorithm~\ref{alg:get} summarizes the steps for \GET, and \Cref{thm:get} shows
the bounds; the steps for
\UPDATE are identical, except that each data value is updated in place.

\begin{tboxalg}{\GET($Q$: batch of query keys)}\label{alg:get}
	\begin{enumerate}[ref={\arabic*}, topsep=0pt,itemsep=0ex,partopsep=0ex,parsep=1ex, leftmargin=*]
    \item Host program finds the PIM module for each operation's key
    using the hash function.
    \item Each PIM module uses its local hash table to find the data node, and return the data value.
\end{enumerate}
\end{tboxalg}

\begin{theorem}
\label{thm:get}
    A batch of \GET (or \UPDATE) operations is executed in one communication
    round, incurring $O(1)$ communication and $O(1)$ expected
    PIM-work per operation.
    The execution is load balanced if the batch size $S = \Omega(P \log P)$.
    It requires $O(S)$ shared-side memory.
\end{theorem}

\begin{proof}
    It takes one communication round, $O(1)$ communication and
    $O(1)$ expected PIM work, because
    \GET (\UPDATE) operations are packed into \textit{constant-size} \GET (\UPDATE) tasks,
    then sent to PIM modules and executed with \textit{constant expected PIM-work} by the
    local hash table of each PIM module.

    Since each \GET (\UPDATE)  incurs constant communication and
    PIM-work, \Cref{lemma:balls_into_bins} proves load
    balance for $S = \Omega(P\log{P})$.
\end{proof}
}

\revisionweak{

\subsection{\INSERT}
\label{sec:insert}

\begin{figure}[t]
    \centering
    \includegraphics[width=1\linewidth]{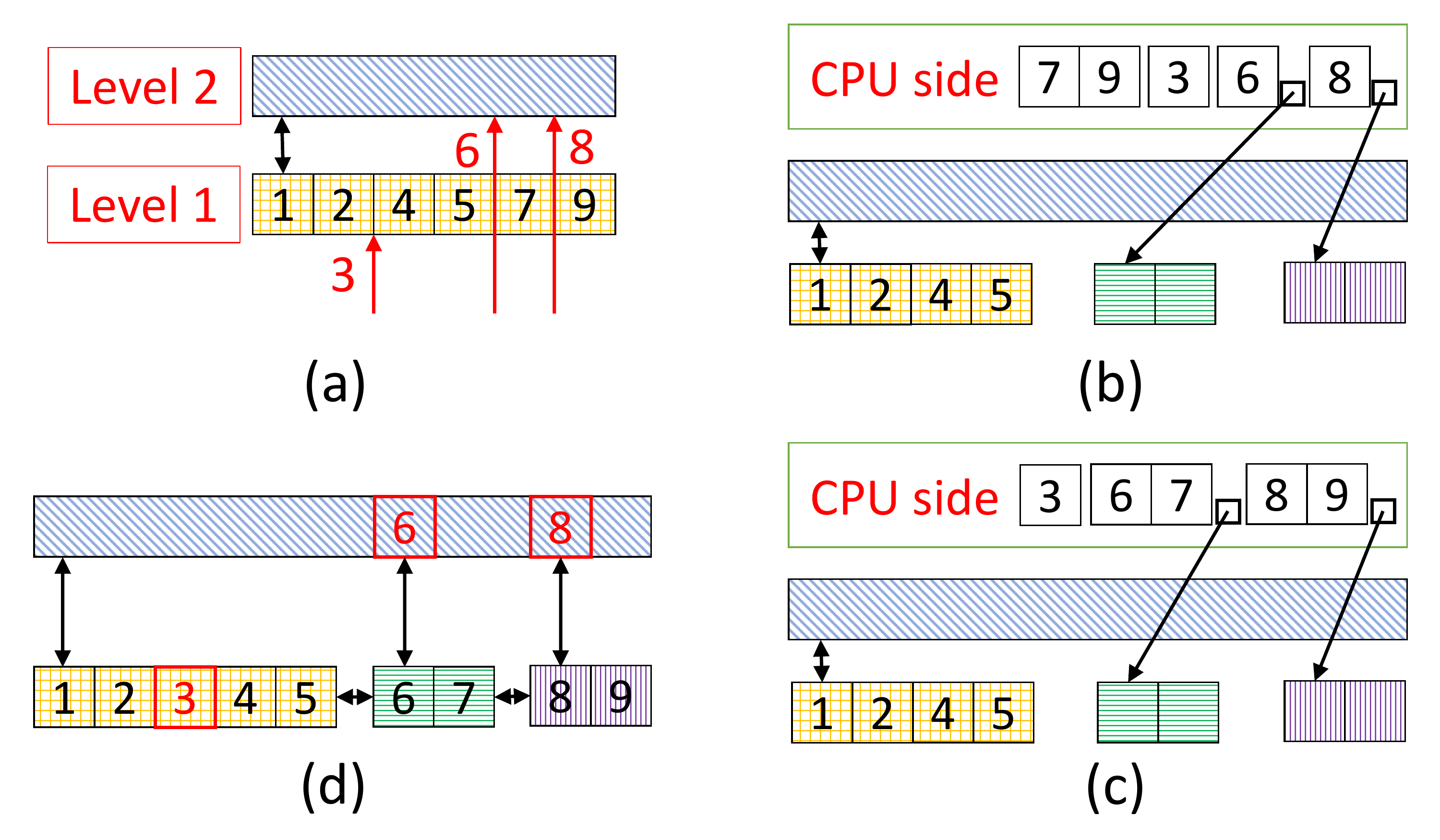}
    \caption{The process to insert keys $3$, $6$, and $8$ into L2 of the PIM-tree. Insertion $3$ has height $1$, and
    insertion $6$ and $8$ both have height $2$. These heights are generated beforehand by coin tossing
    with probability $1/B$. The height of the yellow node is $1$, and that of the blue node is $2$.
    As a result, key $3$ is inserted into the yellow node, and keys $6$ and $8$ split the yellow node.
    }
\label{fig:insert}
    \vspace{-0.07in}
\end{figure}

An \INSERT{(key,value)} operation inserts the key into
nodes on its search path, and the primary challenge is
to avoid contention and conflicts when multiple \INSERT{}s in a batch go
to the same node(s). We solve this by preprocessing:
we perform searches in parallel and record the trace of
each search, and use the traces to detect and handle
contention points. Our algorithm has three stages:
(1)~perform searches
to record the search trace, (2) modify the physical skip list
based on the search trace, and (3) update the shadow subtrees.

\conffulldifferent{
  \myparagraph{Update Physical Skip List.}
  After getting the search traces of each \INSERT,
  we \INSERT into these nodes according to random heights we generate
  prior to updating the PIM-tree.
  \Cref{fig:insert} is an example of our contention solving strategy.
  According to their pre-generated heights, the insertion of $3$ is into the yellow node,
  and the insertions of $6$ and $8$ will split that node. The insertion
  takes three steps: in (b) we fetch
  the right-side part of the node to the CPU side and generate empty new nodes in
  random PIM modules; in (c) we derive the correct element to be inserted to each
  node in CPU; finally in (d) we insert them.

  \hongbo{plural?}
  Choosing \skiplist{}s instead of \bplustrees as the basis of L1 and L2
  helps reduce the number of rounds, since
  we can insert to all nodes in parallel,
  rather than level-by-level bottom up from the leaves like the \bplustree.
  Insert to L3 is also executed in parallel with that of L2,
  performed by broadcasting \INSERT{}s reaching L3 to all PIMs.

  \myparagraph{Update Shadow Subtrees.}
  To maintain the invariant that shadow subtrees are copies of the
  search subtrees, we update the shadow subtrees
  after updating the physical skip list. There are three types of updates:
  (1) \textbf{build} the new shadow subtree for a new node, (2) \textbf{insert}
  a new node into the shadow subtrees of its ancestors, and (3) \textbf{trim} a shadow
  subtree after a node split.

  Our shadow subtree updating technique is straightforward. For build,
  we pull the L2 search tree and send it to
  the new node. For insert and trim, we observe that
  only shadow subtrees of nodes on the search trace need updating, so we
  send the newly-inserted node to all these nodes.

  \myparagraph{Discussion: Load Balance in \INSERT.}
  There is a load balance issue in our shadow subtree update algorithm:
  To keep shadow subtrees up to date, an L2 node
  may need updates of size $O(P/\log_B{P})$. For example, a new L2
  root needs to build its shadow subtree of expected $O(P/\log_B{P})$ nodes
  (given $H'_{L2} = \log_B{P} - \log_B\log_B{P}$).
  This contention factor $O(P/\log_B{P})$ will grow faster than the factor
  $K = B\cdot\log_B{P}$ of \PREDECESSOR as $P$ grows.
  This contention has minor effect at present, but we propose an algorithm
  to address this problem (not implemented at present).

  The solution is not to keep all shadow subtrees up to date, but instead
  to mark some nodes as \textit{unfinished}, update their shadow subtrees gradually in future
  rounds, and avoid using a shadow subtree until it is up to date. This helps
  smooth out the imbalance. \PREDECESSOR bounds
  still hold when the number of such nodes is below a threshold, and we achieve this
  by additional update rounds, which are load balanced when the number of such nodes is high.
}{
  \myparagraph{Record Search Trace.}
  Recording the search trace is simple when there are no shadow
  subtrees, as the addresses of all of the searched nodes correspond to
  their actual (non-replicated) locations. Specifically, we can obtain a
  search trace per-query by recording the PIM pointers used in the
  search process in each Push-Pull step. However, the recording becomes
  challenging when using shadow subtrees, because they enable the search
  to run in the shadow copy (using local pointers), but we still need
  the true physical node addresses to modify the skip list.  To solve this
  problem, we also store the physical pointers in shadow nodes.  Note
  that for each red shadow pointer in \Cref{fig:pim_tree}, there is a
  black physical pointer referencing the physical origin of the shadow
  node.  When the search uses a shadow subtree, the PIM module records
  the corresponding physical pointer, and therefore records the
  correct search trace.

  \myparagraph{Update Physical Skip List.}
  After obtaining the search traces,
  we \INSERT into these nodes according to random heights we generate
  prior to updating the PIM-tree.
  According to the height of each insertion and each layer, all
  insertions go to L1, $1/\log_B{P}$ goes to L2, and $1/P$ goes to L3.  Insertions
  to different layers are applied differently. Insertions to L3 are
  broadcasted to all PIM modules and applied to their L3 copy, which is
  a local \bplustree. For L1 and L2, insertions are applied to nodes on
  the search trace according to their heights. \Cref{fig:insert} is an
  example of insertions with contention,
  where we insert key $3$ into the node as well as splitting the node with
  key $6$ and $8$, according to the pre-generated height.
  The insertion
  takes three steps: in (b) we fetch
  the right-side part of the node to the CPU side and generate empty new nodes in
  random PIM modules; in (c) we derive the correct element to be inserted to each
  node in CPU; finally in (d) we insert them.

  The splitting policy of PIM-tree enables it to insert into different
  levels and different layers at once, rather than level-by-level bottom up
  from the leaves like the \bplustree. Complicated update processes (broadcast in L3,
  fetch and calculation in L2) over a small proportion of data can also be executed
  in parallel with the simple update processes (direct insert in L1) for all data,
  because they don't need results from lower layers.
  In practice, we insert to all levels in $2$
  rounds by (1) initializing new nodes and fetching the right-side
  parts in the first round, then (2) inserting into existing nodes and
  new nodes in the second round.

  Each node receives at most $O(B)$ insertions in these stages, therefore we
  avoid load imbalance. Although the number of concurrent insertions to
  a single shared node isn't limited, these insertions split this node with
  probability $1/B$, and only insertions with keys less than the
  minimum-keyed split will actually be applied to this node, which are $O(B)$
  insertions.

  \myparagraph{Update Shadow Subtrees.}
  To maintain the invariant that shadow subtrees are copies of the
  search subtrees, we need to update shadow subtrees
  after performing insertions to the physical skip
  list. There are three types of updates:
  (1) \textbf{build} the new shadow subtree for a new node, (2) \textbf{insert}
  a new node into the shadow subtrees of its ancestors, and (3) \textbf{trim} a shadow
  subtree after a node split.

  Our shadow subtree updating technique is straightforward. For build,
  we pull the L2 search tree and send it to
  the new node. For insert and trim, we observe that
  only shadow subtrees of nodes on the search trace need updating, so we
  send the newly-inserted node to all these nodes.


  \myparagraph{Algorithm.}
  Algorithm~\ref{alg:insert} summarizes the steps for \INSERT.

  \begin{tboxalg}{\INSERT($K$: batch of keys, $V$: batch of values)}\label{alg:insert}
    \begin{enumerate}[ref={\arabic*}, topsep=0pt,itemsep=0ex,partopsep=0ex,parsep=1ex, leftmargin=*]
      \item Generate a height for each \INSERT according to the geometric distribution with probability $1/B$.
      \item Run \PREDECESSOR on the keys to obtain/record search traces.
      \item Distribute new empty nodes to PIM modules, and record the address of these nodes in the CPU side.
      \item For each node that is split, send the minimum splitting key to fetch the right-side part.
      \item On the CPU side, compute the contents of each new node, and
      the new horizontal pointers.
      \item Apply all insertions to L1 and L2 nodes for both old and new nodes, and build horizontal pointers.
      \item Broadcast insertions that reach L3 to all PIM modules.
      \item Update shadow subtrees in existing L2 nodes.
      \item Build shadow subtrees for newly inserted L2 nodes.
  \end{enumerate}
  \end{tboxalg}
  Note that steps 3--4 and 6--8 are both done using a single communication round.

  \myparagraph{Optimizations.}
  We don't need the whole search path for most insertions.
  First, we only record the search
  trace in L1 and L2, as insertions to L3 are executed locally by
  the PIM modules. Also, L1 insertions with height $h < H_{L1}$ only need and
  affect $h$ nodes at the bottom of its search trace, therefore we only
  record these nodes, which is $O(1)$ nodes in expectation.
  This optimization
  ensures that each insertion records $O(1)$ nodes in their search path on
  average, because L1 insertions need only $O(1)$ expected nodes, and
  only a $1/\log_B{P}$-fraction of the insertions reaching L2
  require a search trace with $\log_B{P}$ nodes.

  \myparagraph{Discussion: Load Balance in \INSERT.}
  There is a load balance issue in our shadow subtree update algorithm:
  To keep shadow subtrees up to date, an L2 node
  may need updates of size $O(P/\log_B{P})$.
  First, new subtrees can have up to $O(P/\log_B{P})$
  (given $H'_{L2} = \log_B{P} - \log_B\log_B{P}$);
  second, a single shadow subtree may receive up to $O(P/\log_B{P})$ updates.
  (We derive this bound in a similar way to the concurrent insertions
  to a single node. An insertion splits the shadow subtree if it reaches L3.
  Each insertions reaching L2 splits the shadow subtree with probability $\log_B{P}/P$,
  so the number of actual insertions to a shadow subtree is $O(P/\log_B{P})$).
  This contention factor $O(P/\log_B{P})$ will grow faster than the factor
  $K = B\cdot\log_B{P}$ of \PREDECESSOR as $P$ grows.
  This contention has minor effect at present,
  but it may cause issues when $P$ gets larger in the future.
  Here we propose an algorithm to address this problem
  (not implemented at present).




  A solution to avoiding load imbalance
  caused by contention, even as the number of PIM modules scales,
  is to delay the update process. To
  identify contented shadow subtrees, we keep track of the number of
  shadow-subtree updates of any existing node, as well as the unbuilt
  shadow-subtree sizes for new nodes.  When the unfinished work on any
  node exceeds a threshold, applying the pending updates will cause
  contention, which leads to load imbalance.
  Therefore, instead of fully updating this node in a single round, we
  mark the node as \textit{unfinished}, and prevent future queries from calling
  \Push on it to avoid using its unfinished shadow subtree.  Whenever
  the overall number of unfinished nodes in the entire tree
  reaches a threshold $K_i$, we start an update phase of
  multiple update rounds, where we send constant information to each
  contention point, until the number of unfinished nodes drops below $K_i$.
  According to \Cref{lemma:balls_into_bins}, $K_i = P\log{P}$
  suffices to ensure the communication in each update round is balanced.

  Unfinished nodes bring only one modification to the \PREDECESSOR algorithm:
  In the multi-round \PullOnly phase, we have to \Pull any unfinished nodes
  on the search path of any query.
  \begin{theorem}
      \Cref{thm:pred} still holds with unfinished nodes.
  \end{theorem}

  \begin{proof}
      We only need to prove for \Cref{lemma:pull_only_proof} because
      we only change the \PullOnly stage. The lemma still holds
      because it takes only $O(B \log{P})$
      additional communication to \Pull
      at most $K_i = P\log{P}$ unfinished nodes in each round, and
      $O(B \log{P}) \le O(\frac{S}{H'_{L2} \cdot P})$.
  \end{proof}

  \begin{lemma}
    \label{lemma:constant_modify}
    An \INSERT only modify $O(1)$ expected nodes.
  \end{lemma}

  \begin{proof}
    An \INSERT affects the PIM-tree differently according to which (L1, L2, L3)
    layer it reaches. This is decided by the pregenerated height.

    All \INSERT{}s affects L1.
    An L1-only \INSERT affects $O(1)$ expected nodes because its height is $O(1)$ expected.
    $O(1/\log_B{P})$ \INSERT{}s are L2 \INSERT{}s.
    An L2 \INSERT affects $H_{L1} + H_{L2} = \log_B{P}$ nodes, because
    it needs to update all shadow subtrees of its L2 ancestors.
    $O(1/P)$ \INSERT{}s are L3 \INSERT{}s.
    An L3 \INSERT affects $O(P)$ nodes because we need to update all replicas
    of the modified L3 nodes.
  \end{proof}
}

\confversiononly{\vspace{-0.1in}}
  \begin{theorem}\label{thm:insert}
    A batch of \INSERT operations can be executed in $O(\log_B{P})$
    IO rounds, incurring $O(\log_B \log_B{P})$ communication for each operation.
    The execution is load balanced if the batch
    size $\batch = \Omega(P\log{P}\cdot B\cdot H_{\text{L2}}') =
    \Omega(P\log{P}\cdot B\cdot \log_B{P})$.
    The CPU-side memory footprint is $O(S)$.
  \end{theorem}
\confversiononly{\vspace{-0.13in}}

\fullversiononly{
  \begin{proof}
    We treat an \INSERT execution as a \PREDECESSOR and additional steps. As the bounds
    for \INSERT is the same as that of \PREDECESSOR in IO rounds, communication and
    batch size requirements, here we prove that the additional steps also follows this bound.

    We first prove for the $O(S)$ CPU-side memory footprint. As we've already
    proved $O(S)$ memory footprint for \PREDECESSOR in \Cref{sec:proof_for_tree_search},
    we only prove that the additional memory space required by \INSERT is $O(S)$.
    These memory is mainly comprised of search paths and new nodes.
    (1) A batch of $S$ \INSERT{}s
    needs $O(S)$ \whp addresses for search paths, because we only need the addresses
    of the affected nodes, which is $O(1)$ for each \INSERT according
    to \Cref{lemma:constant_modify}. (2) For new nodes,
    $O(S/B)$ \whp new nodes are generated, and the size of each
    follows a geometric distribution with expectation $B$, given a
    $O(S)$ \whp total size.

    We then prove for the $O(\log_B\log_B{P})$ communication and $O(\log_B{P})$
    IO rounds. The additional steps take $O(1)$ IO rounds. For communication, each
    operation causes $O(1)$ communication because the costs of both inserting into
    existing nodes and generating new nodes are $O(1)$ expected for each node.
    Amortizing over a batch gives $O(\log_B\log_B{P})$ \whp communication.

    We then prove for the $O(\log_B\log_B{P})$ communication, $O(\log_B{P})$
    IO rounds, and load balance for the additinoal steps.
    step 1 and 5 causes no communication; step 7 is a broadcast of $O(S/P)$ data,
    hence causes load-balanced $O(S)$ communication; step 3 distributes
    $O(S/B) = O(P\log{P}\log_B{P})$ unit-size empty nodes each to a random PIM module.

    In step 4, we pull the right-side part of
    $O(S/B) = O(P\log{P}\log_B{P})$ nodes \whp, each node's size
    following geometric distribution with expected value $B$.
    The proof for $O(S/P)$ IO time is similar to that in \Cref{lemma:pull}.

    In step 6, we prove $O(S/P)$ IO time seperately for two type of insertions:
    building new nodes and insert into existing nodes.
    The proof of $O(S/P)$ IO time for new nodes is similar to that of step 4,
    as we build $O(S/B)$ new nodes with the same size distribution.
    Bounds for insertion into non-leaf existing nodes can be proved similarly,
    as there's $O(S/B)$ such nodes. For insertion to leaf existing nodes,
    there is (1) up to $S$ such insert; (2) up to $S$ target leaf existing nodes
    for uniform random workload;
    (3) up to $B$ \whp insert to a single node in case of data skew
    (geometric distribution with expectation $B$).
    We divide target leaf nodes into two types by whether this node has more than
    $B$ insertions. For nodes with less than $B$ insertion,
    we apply \Cref{lemma:weightbalanced}, where the weight limit is $B$ and the
    total weight is the total number of insert (no more than $S$).
    For nodes with more than $B$ insertion, there is no more than $(S/B)$ such node,
    and the number of additional inserts (other than the first $B$ inserts)
    for each node follows geometric distribution with possibility $1/B$.
    so we apply \Cref{lemma:balls_into_bins} and \Cref{lemma:geometric_balanced}
    respectively for the first $B$ inserts and additional inserts, to prove $O(S/P)$
    IO time.

    For shadow subtrees, we proof the sum of updates to all shadow subtrees is $O(S)$.
    We don't prove for for load balance, becaused they're executed by the
    update rounds that ensure load balance. We prove seperately for the two types of updates:
    updates to existing shadow subtrees, and building new shadow subtrees.
    There's $O(S)$ \whp updates to existing shadow subtrees, because $O(S/log_B{P})$ \whp insertions
    reach L2, and each insertion does $O(1)$ update to less than $\log_B{P}$ shadow subtrees
    (of its L2 ancestors).
    Building new shadow subtrees also take $O(S)$ updates: we denote L2 leaves as level 0,
    the L2 includes $\log_B{P} - \log_B\log_B{P} < \log_B{P}$ levels. Since the probability
    of PIM-tree is $1/B$,
    there's $O(S/\log{P}/B^{i + 1})$ \whp new nodes generated on level $i$, and we need
    to build a shadow subtree with size $O(B^{i + 1})$ \whp for each node. Multiplying the number
    of new nodes on each level and the size of these nodes gives $O(S/\log{P})$ updates on each
    level, and $O(S)$ updates for all levels.
  \end{proof}
}


\conffulldifferent{
  \myparagraph{Implementing \DELETE}.
  We handle deletions in similarly to insertions: first obtain the
  search trace, then delete keys from nodes on the trace, finally
  apply updates to shadow subtrees. While insertion causes node split,
  deletion causes nodes to merge when removing the pivot key from a
  node. See the full paper~\cite{pimtree2022full} for details.
}{
\subsection{\DELETE}
We perform deletions in a similar way to insertions:
first get the search trace, then delete keys from nodes on the search trace, and finally apply updates
to shadow subtrees. While insertion causes node split, deletion causes node merge when the pivot key is deleted from a node.


We need the height of each key to modify the tree. While these heights are pre-generated when doing inserts,
we need to get them before doing deletions. We store the height of each key in the data node.
Before deletions, we do a batch \GET to get these heights and filter out invalid deletions.
With these heights, we do deletions just like insertions: modify the physical tree by removing keys and
merging nodes, then update shadow subtrees for nodes on the search path.

A batch of \DELETE operations can be executed in $O(\log_B{P})$ communication rounds,
incurring $O(\log_B \log_B{P})$ communication for each operation.
Algorithm~\ref{alg:delete} summarizes the steps.

\begin{tboxalg}{\DELETE($Q$: batch of query keys)}\label{alg:delete}
	\begin{enumerate}[ref={\arabic*}, topsep=0pt,itemsep=0ex,partopsep=0ex,parsep=1ex, leftmargin=*]
    \item Preprocess: get the height for each valid \DELETE, and remove any invalid \DELETE{}s.
    \item Batch search to obtain search traces for each key.
    \item Remove the deleted nodes from the search traces.
    \item Fetch the remaining part if a node's pivot key is deleted.
    \item Merge horizontal consecutive remaining parts in the CPU.
    \item Insert the remaining parts to the left-side node of the removed nodes.
    \item Build horizontal pointers.
    \item Broadcast deletions that reach L3 to all PIM modules.
    \item Update shadow subtrees in existing L2 nodes.
\end{enumerate}
\end{tboxalg}
Note that steps 3--4 and 6--8 are both done using a single communication round.
}
}

\subsection{\SCAN}

The \SCAN{(LKey,Rkey)} operation (a.k.a.~range query) returns all the
(key, value) pairs whose keys fall into the range of [Lkey, Rkey].
\conffulldifferent{
  \revision{Its algorithm is similar to \PREDECESSOR.}
}{
  Algorithm \ref{alg:scan} summarizes the \SCAN process.
}

When running a batch of \SCAN queries,
we first on the CPU side merge all overlapping ranges in the batch into
groups of disjoint larger \conffulldifferent{ranges.}{
  ranges, by sorting the Lkey of the ranges and
  then carrying out a prefix sum on the Rkey, with the binary associative
  operator set to be $max()$.
  The leftmost bound of these merged disjoint ranges are the Lkey$[i]$ where
  $\mbox{Lkey}[i]>\mbox{Rkey}[i-1]$.
  PIM-tree can use user-defined thresholds to split too large merged ranges
  into several small disjoint ranges to avoid overflows on hardware as well
  as achieve load-balance.
  Then the PIM-tree carries out batched scans on these disjoint ranges and
  eventually rearranges corresponding results from the fetched (key, value) pairs.
}
PIM-tree then carries out \SCAN throughout the L3 traversal, by evenly distributing the batched range quries to all PIM modules and maintaining two boundary nodes (the predecessors of Lkey and Rkey on the current level) for each range query through the level-by-level L3 search.
Within a single range, only the two boundary nodes and their intermediates will be involved later.
The two boundary nodes are marked with \SearchRequired labels, while
the intermediates are marked with \FetchAll labels.

\FetchAll nodes are required to return all their leaf data nodes.
Note that the range queries sent to the PIM modules are disjoint, so \FetchAll nodes do not generate contention points throughout the L2 searching.
Thus, \FetchAll queries can be simply pushed to PIM modules and take advantage of shadow subtrees.

\conffulldifferent{
The \SearchRequired nodes are processed similar to \PREDECESSOR,
using \PushPullSearch.
Contention points are pulled to the CPU, while the others are pushed to the PIM modules.
}{
The \SearchRequired nodes are processed similar to \PREDECESSOR.
Since \SearchRequired nodes might overlap on some highest levels in L2 and thus cause load-imbalance, similar \PushPullSearch is taken.
\SearchRequired nodes with large contentions are pulled to the CPU side, while the others are pushed to the PIM modules.
}
\pimtree{} maintains the two boundary nodes in each range query by a \PREDECESSOR-like searching throughout L2, while the newly-generated intermediate nodes in lower levels are labeled with \FetchAll.

\fullversiononly{
\begin{tboxalg}{\SCAN($R$: batch of range queries)}\label{alg:scan}
	\begin{enumerate}[ref={\arabic*}, topsep=0pt,itemsep=0ex,partopsep=0ex,parsep=1ex, leftmargin=*]

    \item Preprocessing: On the CPU side, merge the overlapping range queries in $R$; then split too large ranges in the merged results.\label{scan:preprocess}
    \begin{enumerate}[ref={\arabic*}, topsep=0pt,itemsep=0ex,partopsep=0ex,parsep=1ex, leftmargin=*]
      \item Sort $R$ with the $LKeys$ in ascending order. Store the sorted ranges in $R'$.
      \item Construct a new array $RRkeys$ with $RRkeys[i]$ representing the maximum value in $R'.RKeys[0:i]$, by using a parallel prefix sum on $R'.RKeys$ with the binary associative operator set to be $max()$.
      \item Each $i$ with $R'.LKeys[i]>RRkeys[i-1]$ represents a starting of a merged range required. Construct merged range queries $R''$ based on a parallel comparison and packing.
      \item Split too large merged ranges in $R''$ based on user-defined thresholds.
    \end{enumerate}

    \item \Push range queries in $R''$ evenly to the PIM modules and traverse L3. For L3 nodes $N_1, N_2, ..., N_m$ that falls into a single range query, mark $N_1$ and $N_m$ with \SearchRequired labels, and the others with \FetchAll labels. \label{scan:L3}

    \item Return all leaf nodes for \FetchAll queries.\label{scan:fetch}
    \begin{enumerate}[ref={\arabic*}, topsep=0pt,itemsep=0ex,partopsep=0ex,parsep=1ex, leftmargin=*]
      \item Push each \FetchAll query to the PIM module holding its search node, and traverse L2 using the shadow subtrees.
      \item Traverse L1 and retrieve the data nodes.
    \end{enumerate}

    \item Use push-pull methods to process \SearchRequired nodes with different contentions.\label{scan:search}
    \begin{enumerate}[ref={\arabic*}, topsep=0pt,itemsep=0ex,partopsep=0ex,parsep=1ex, leftmargin=*]
      \item \Pull all nodes with more than $K = B\cdot H_{\text{L2}}'$
        queries back to the CPU. Use these nodes to process the
        pointer-chasing process of these queries by one step.
      \item \Push other nodes to the PIM module. Traverse L2 using the shadow subtrees. Maintain two boundary nodes in a range as \SearchRequired through all levels of the PIM-tree and perform \PREDECESSOR-like searches until retrieving the data nodes. Mark the intermediate nodes as \FetchAll and process them using Step \ref{scan:fetch}.
    \end{enumerate}

    \item Rearrange the result with $R$ and the returned key-value pairs on the CPU side.

	\end{enumerate}
\end{tboxalg}
}

\section{Implementation}

\myparagraph{CPU-PIM Pipelining.}\label{sec:pipelining}
Thus far, we have introduced algorithms where tasks on the CPU and PIM
run in a synchronized, tick-tock manner in each round as depicted in
Algorithm~\ref{alg:pattern}.  The total execution time of this
approach consists of three non-overlapping components: CPU-only time,
PIM-only time, and communication time.  Communication requires both
CPU and PIM, but the other two components only utilize one part of the
system, which presents an opportunity to reduce execution time by
pipelining the CPU-only and PIM-only components.

For pipelining, we consider executions that run multiple batches in
parallel in the PIM-tree. 
This is shown in \Cref{fig:cpu_pim_pipelining},
where ``CPU'' represents time spent in CPU-only execution,
while ``IO \& PIM'' represents time spent in CPU-PIM communication and the PIM program.
On our UPMEM system,
CPU-PIM communication requires exclusive control of the PIM side,
and any concurrent use of the PIM side will cause a hardware fault.
Hence, one batch needs to wait for the PIM side to finish the current execution tasks.
We only pipeline queries in our experiments, since update batches cannot be carried out concurrently. 
For mixed operations, we protect the PIM-tree by a read-write lock to
prevent update batches from running concurrently with other batches.

\begin{figure}[t]
    \centering
    \includegraphics[width=0.8\linewidth]{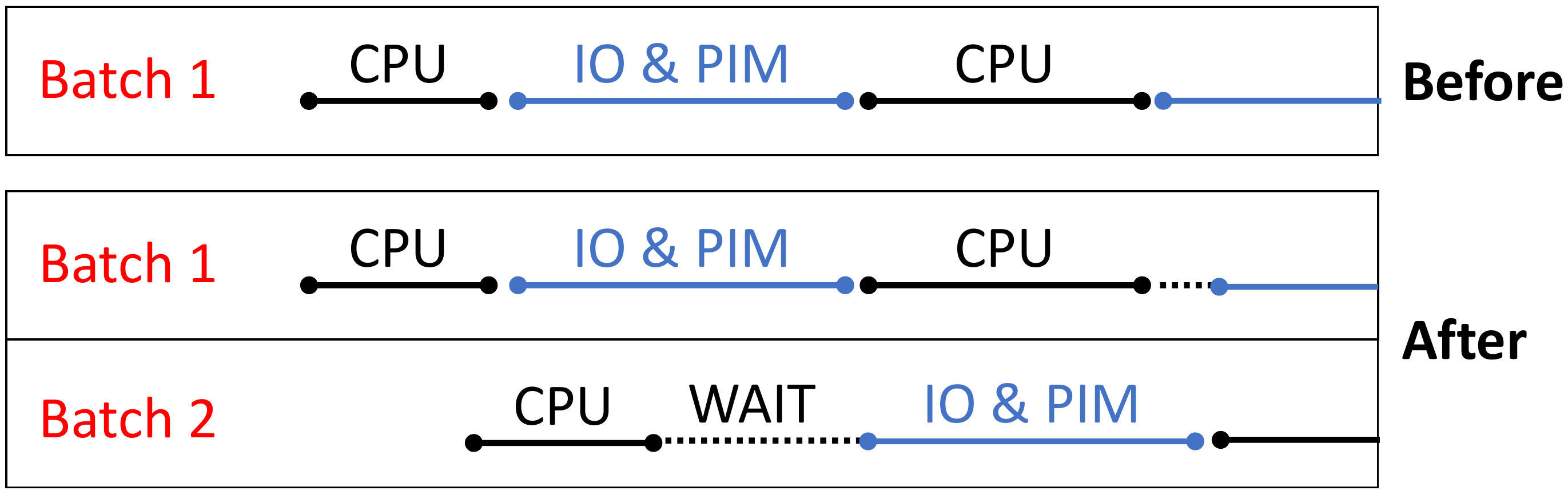}
    \caption{Program traces before/after CPU-PIM pipelining
    }
\label{fig:cpu_pim_pipelining}
\end{figure}

\myparagraph{PIM Program.}
\label{sec:pim_program}
PIM-tree's PIM program is a parallel executor of the tasks in the
buffer sent from the CPU.
It is designed to address two features of UPMEM's current PIM processors.
First, the PIM processor is a fine-grained multi-threaded computing
unit~\cite{gomez2021benchmarking}, and requires at least 11 threads to
fill the pipeline, so we write PIM programs in the form of 12 threads.
Second, UPMEM's system only supports PIM programs with fewer than 4K
instructions, but the implementation of PIM-tree exceeds this bound.
To bypass this restriction, we write the PIM program as multiple separate
modules, and load each module when needed. Only \INSERT and \DELETE
operations require swapping modules; program loading currently takes around
$25\%$ of the execution time. The remaining
operations on PIM-tree fit within the 4K instruction limit.



\section{Evaluation}


\revision{
In this section, we evaluate our new PIM-optimized indexes on a PIM-equipped
machine provided by UPMEM, and two traditional state-of-the-art indexes on a machine
with similar performance. We summarize our expeirmental results from
this section as follows:

\begin{enumerate}[leftmargin=*]
  \item The PIM-tree performs better than the range-partitioned skip list
        under skewed workloads in terms of throughput, memory-bus
        communication, and energy consumption.
  \item The PIM-tree causes lower communication on the memory bus compared
        with traditional indexes without PIM.
  \item All optimizations mentioned, including \PushPullSearch, shadow
  subtrees, chunked skip list and CPU-PIM pipelining, yield
  performance increases to (some) PIM-tree operations.
\end{enumerate}
}

\subsection{Experiment Setup}

\myparagraph{UPMEM's PIM Platform}.
We evaluate PIM-tree on a PIM-equipped server provided by UPMEM(R).
The server has two Intel(R) Xeon(R) Silver 4126 CPUs, each CPU with 16
cores at 2.10 GHz and $22$ MB cache.
Each socket has 6 memory channels: $4$ DIMMs of conventional DRAM are
implemented on $2$ channels, while $8$ UPMEM DIMMs are on the other
$4$ channels.
Each of the 16 UPMEM DIMMs has 2 ranks, each rank has 8 chips, and
each chip has 8 PIM modules.  \revision{There are \numDPUsAll PIM modules in
total.}

\revision{
\myparagraph{Traditional Machine w/o PIM}.
We evaluate traditional indexes on a machine with
two Intel(R) Xeon(R) CPU E5-2630 v4 CPUs, each CPU with 10 cores at
2.20 GHz and $25$ MB cache. Each socket has $4$ memory channels. There
are no PIM-equipped DIMMs.
We cannot evaluate traditional indexes on the server of UPMEM because
$2/3$ of its memory channels are used by PIM-equipped DIMMs, which cannot
be used as the main memory. Directly running traditional indexes on the
server would cause unfairness in main memory bandwidth for the
traditional indexes.
In our experiments we choose the state-of-the-art
\bst \cite{DBLP:conf/ppopp/BronsonCCO10}
and \abtree \cite{brown2017thesis} as competitors. Both implementations
are obtained from
\conffulldifferent{SetBench~\cite{DBLP:conf/usenix/Arbel-Raviv0018}.}
{the SetBench benchmarking suite~\cite{DBLP:conf/usenix/Arbel-Raviv0018}.}
}



\confversiononly{\vspace{-0.3cm}}
\myparagraph{Range-Partitioned and Jump-Push Baselines}.
We implement a range-partitioned-based ordered PIM index as our primary baseline,
where both data nodes and index nodes are distributed to PIM modules based on the ranges of the key~\cite{liu2017concurrent, choe2019ndp}.
We record the range splits in the CPU side, and use these splits to find the targeted PIM module
of each operation. Point operations are sent to and executed on the corresponding PIM module.
Running a batch of \SCAN operations is similar, except that it runs an additional
splitting in queries according to the range splits before tasks are sent to the
PIMs.
We also build a local hash table on all PIM modules for \GET.

\revision{
We also implement the PIM-balanced skip list~\cite{kang2021processing}
described in \S\ref{sec:pim_balanced_index} as another baseline.  We
experimentally evaluate this approach when discussing the impact of
the optimizations proposed in this paper,
in \Cref{fig:pim_tree_optimizations}, where the algorithm is called ``Jump-Push based''.
}

\revision{
\confversiononly{\vspace{-0.3cm}}
\myparagraph{Test Framework.}
We run multiple types of operations on the PIM-tree, range-partitioning
skip list we implemented, and the state-of-the-art traditional indexes.
In all experiments, we first warm up the index by running the
\textbf{initialize set} that \INSERT key-value pairs,
then evaluate the index by the \textbf{evaluation set} of multiple operations.
All operations are loaded from pre-generated test files. PIM algorithms
(the PIM-tree and range-partitioning) run operations in batches, and traditional
indexes run them directly with multi-threaded parallelism.
In all experiments, the sizes of both keys and values are set to 8 bytes.

To study the algorithms, we measure both the time spent, and
the memory bus traffic. Memory bus traffic is measured by adding CPU-PIM and
CPU-DRAM communication, the prior one measured by a counter increased
whenever a PIM function(e.g., PIM\_Broadcast) is called,
and the later one measured as cache misses by PAPI.
We bind the program to a single NUMA node and disable the CPU-PIM pipeline
when measuring cache misses for an accurate traffic measurement.
As each CPU of the PIM-equipped machine has two NUMA nodes,
the effective cache of the PIM algorithms is reduced to $11$ MB,
half of the full cache, under this setup.
Time is measured with full interleave over all NUMA nodes.

Instability in performance exists in the current generation of PIM hardwares.
We oberserved an approximately $\pm15\%$ fluctuation in the measured metrics mentioned above in our experiments.
}

\subsection{Microbenchmarks}
\label{sec:overall_performance}

\myparagraph{Workload Setup}.
\label{sec:micro_setup} Each test first warms up the index by inserting
\dataSize uniform random key-value pairs;\footnote{This is a favorable setting
for the range-partitioned baseline, because
the range boundaries are stable. The performance of the PIM-tree
is not impacted by the distribution of key-value pairs over time.}
then for testing it executes (i) \testSize point operations or (ii)
\scanTestSize \SCAN operations that each retrieve $100$ elements in expectation.
Point operations use batch size $S=\mbox{\batchSize}$, and \SCAN operations
use batch size $S=\mbox{\scanBatchSize}$.

We generate skewed workloads with Zipfian
distribution~\cite{zipf2016human}.  However, workloads generated by
Zipf-skew over elements is not ideal for evaluating batch-parallel
ordered indexes, because this skew can be easily handled by a
deduplication in preprocessing on the CPU side, by merging operations
of the same key into one.
\revisionweak{To better represent the spatial bias, where
keys in some ranges are more likely to be accessed in the same batch,
we slightly modify the way to generate our Zipfian workload, as follows:
(i) we divide the key space evenly into $P = \numDPUs$ parts;
(ii) for each operation, we first choose a part according to the
Zipfian distribution, then choose a uniformly random element in that
part.}  For operations for existing keys (\GET, \DELETE), we divide and choose among
the keys currently in the index; for operations on
arbitrary keys (\PREDECESSOR, \INSERT, \SCAN), the key space consists
of all 64-bit integers.
We periodically shuffle the probability of
each part in Zipfian distribution.  This helps alleviate, but not
eliminate the PIM memory overflow problem of the range-partitioned
baseline caused by \INSERT{}s accumulating in high-probability parts.
PIM-tree gains no benefit from this shuffle.

To show results on different amounts of skew, we evaluate
the algorithms on different $\mathbf{\alpha}$ values in the
Zipfian distribution, ranging from $0$ (uniformly random) to $1.2$.
\revision{With this skewness generation approach,
less than $10\%$ operations are eliminated because of duplicated keys,
under the most skewed case ($\mathbf{\alpha} = 1.2$).}

%
%

\begin{figure*}
  \vspace{-5.5em}
  \centering
  \subcaptionbox{Get (by hash)\label{fig:get_performance}}
                {\vspace{-0.8em}\includegraphics[width=0.3\linewidth]{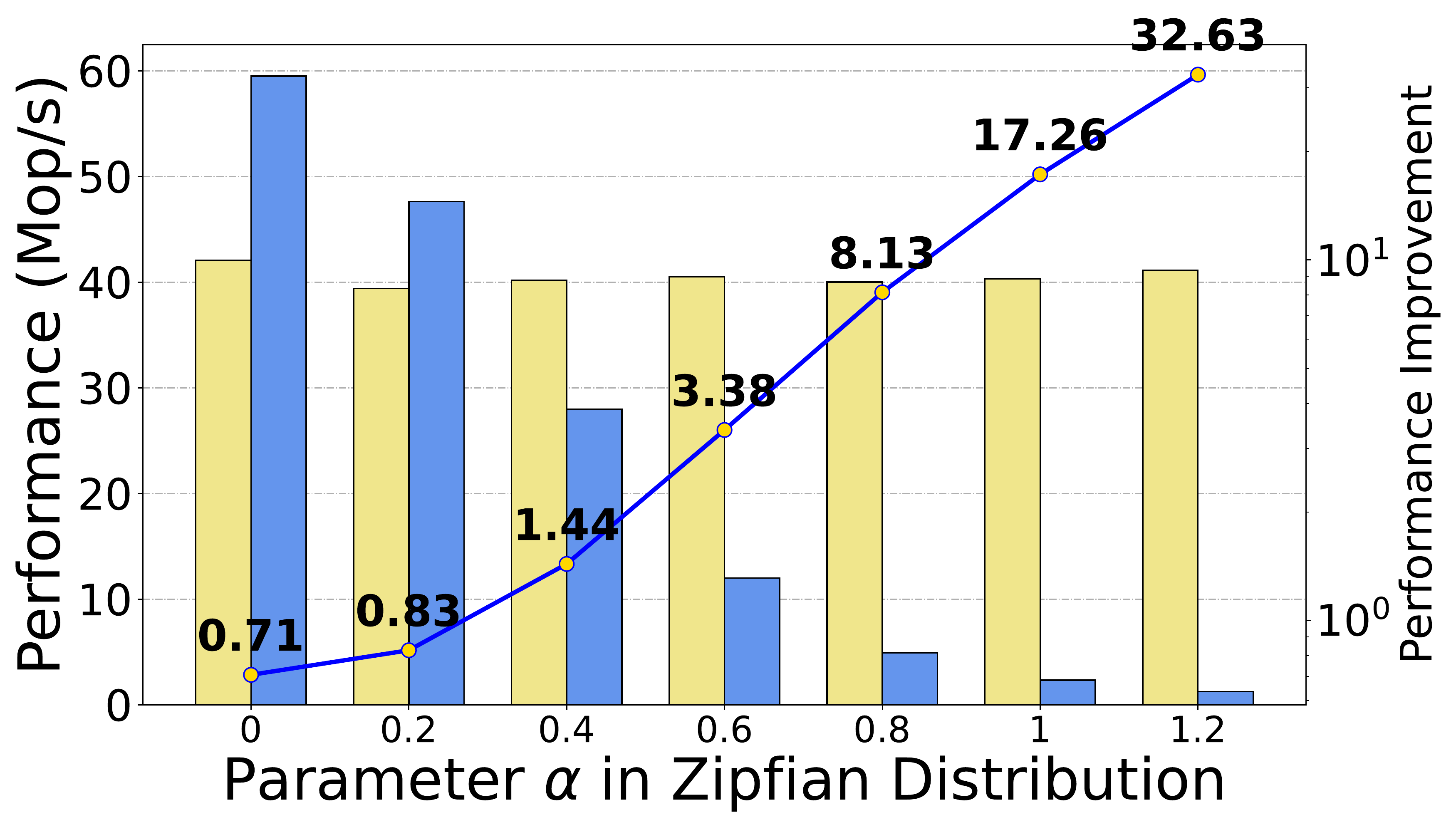}}
                \hspace{0.05in}
  \subcaptionbox{Predecessor\label{fig:search_performance}}
    {\vspace{-0.8em}\includegraphics[width=0.3\linewidth]{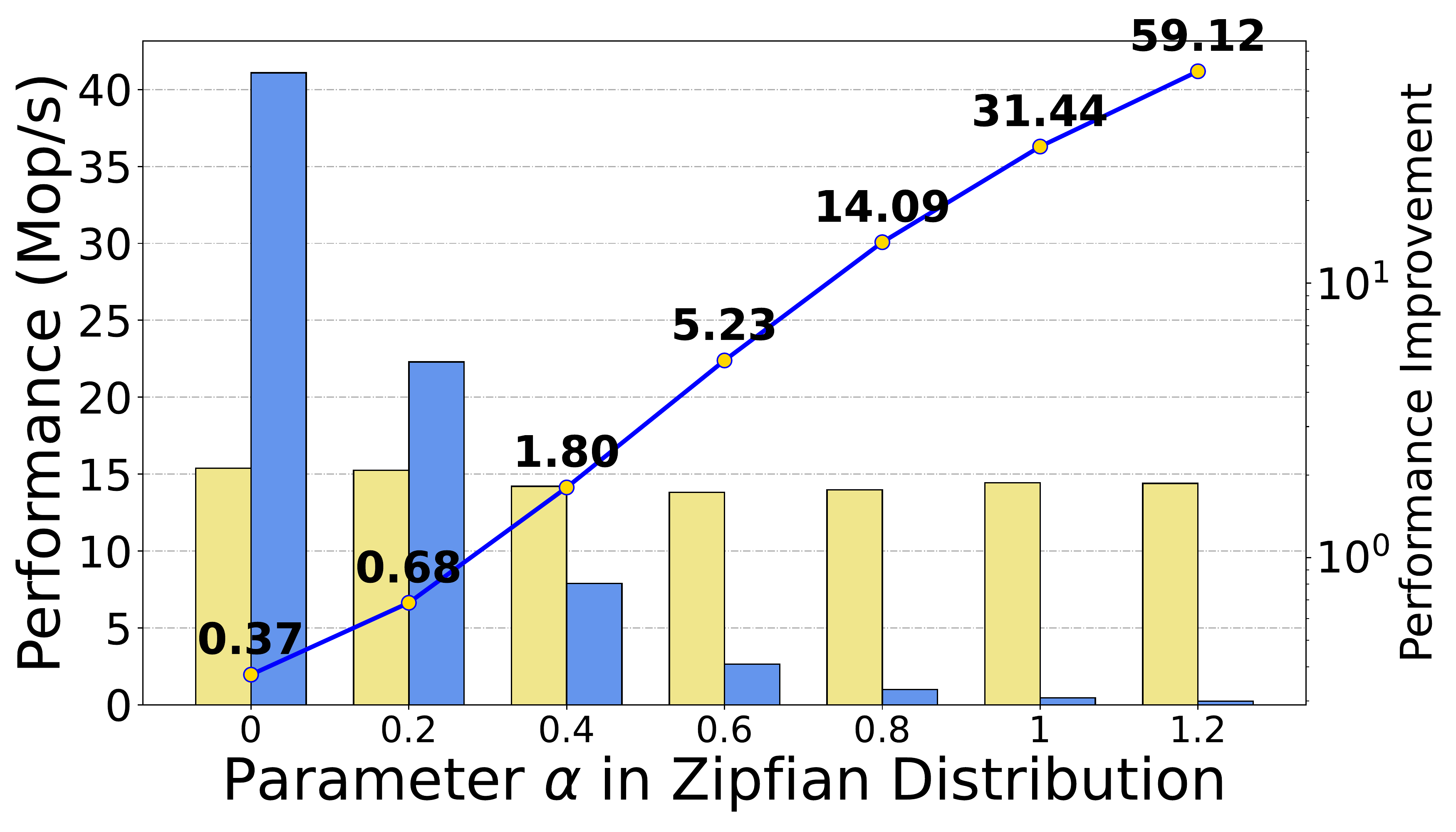}}
                \hspace{0.05in}
  \subcaptionbox{Insert\label{fig:insert_performance}}
    {\vspace{-0.8em}\includegraphics[width=0.3\linewidth]{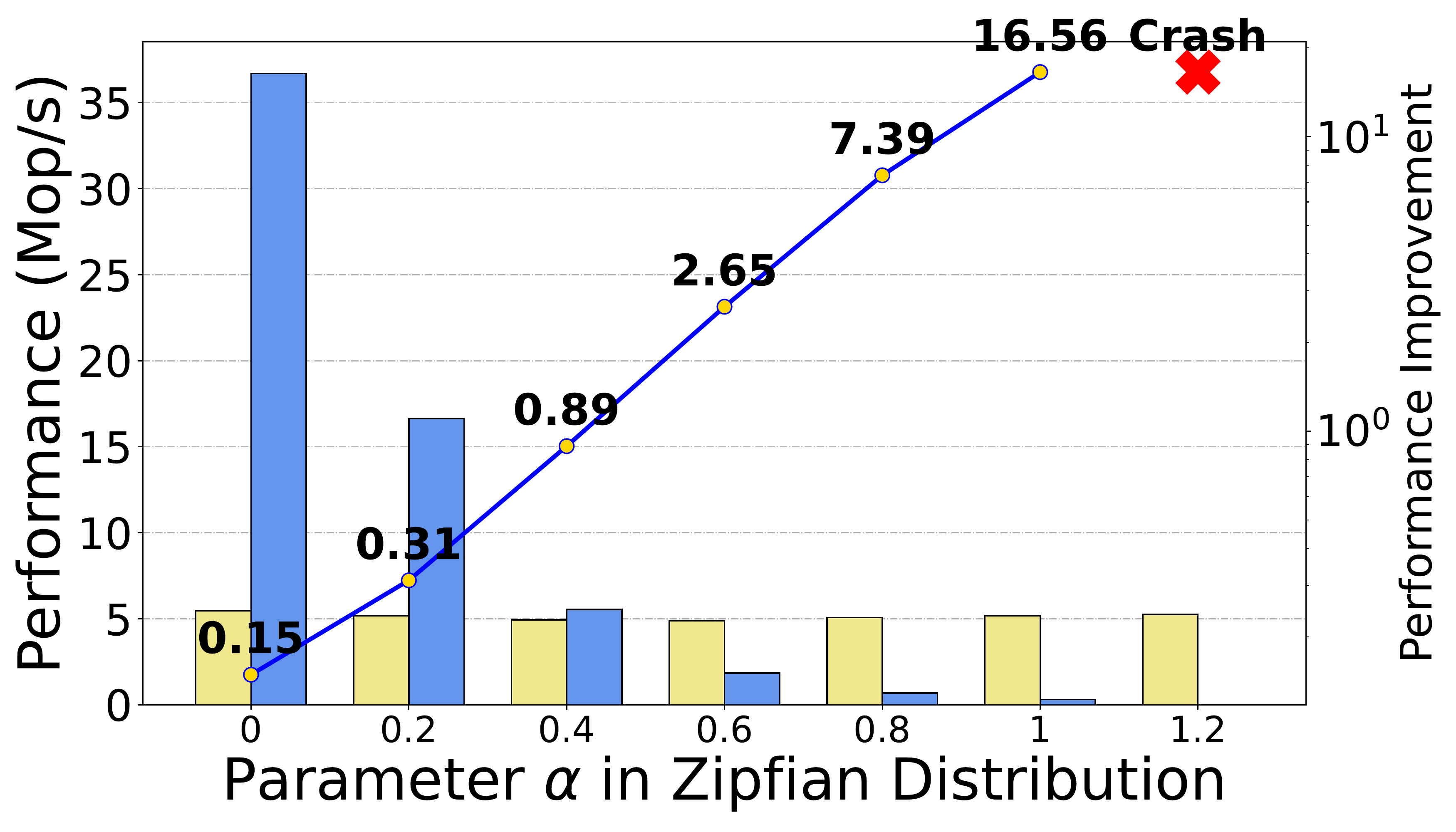}}
  \subcaptionbox{Delete\label{fig:delete_performance}}
    {\vspace{-0.8em}\includegraphics[width=0.3\linewidth]{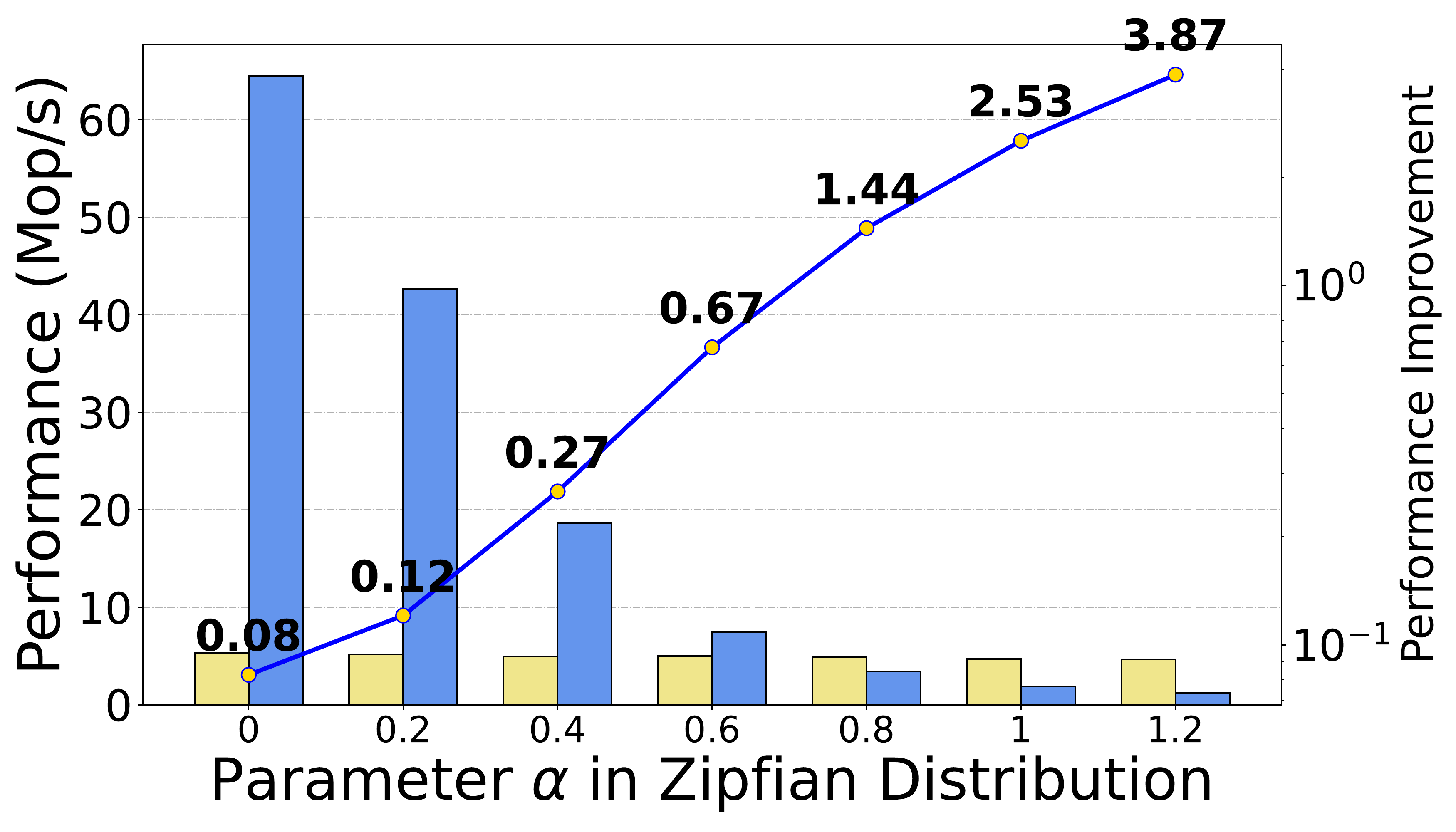}}
                \hspace{0.05in}
  \subcaptionbox{Scan\label{fig:scan_performance}}
    {\vspace{-0.8em}\includegraphics[width=0.3\linewidth]{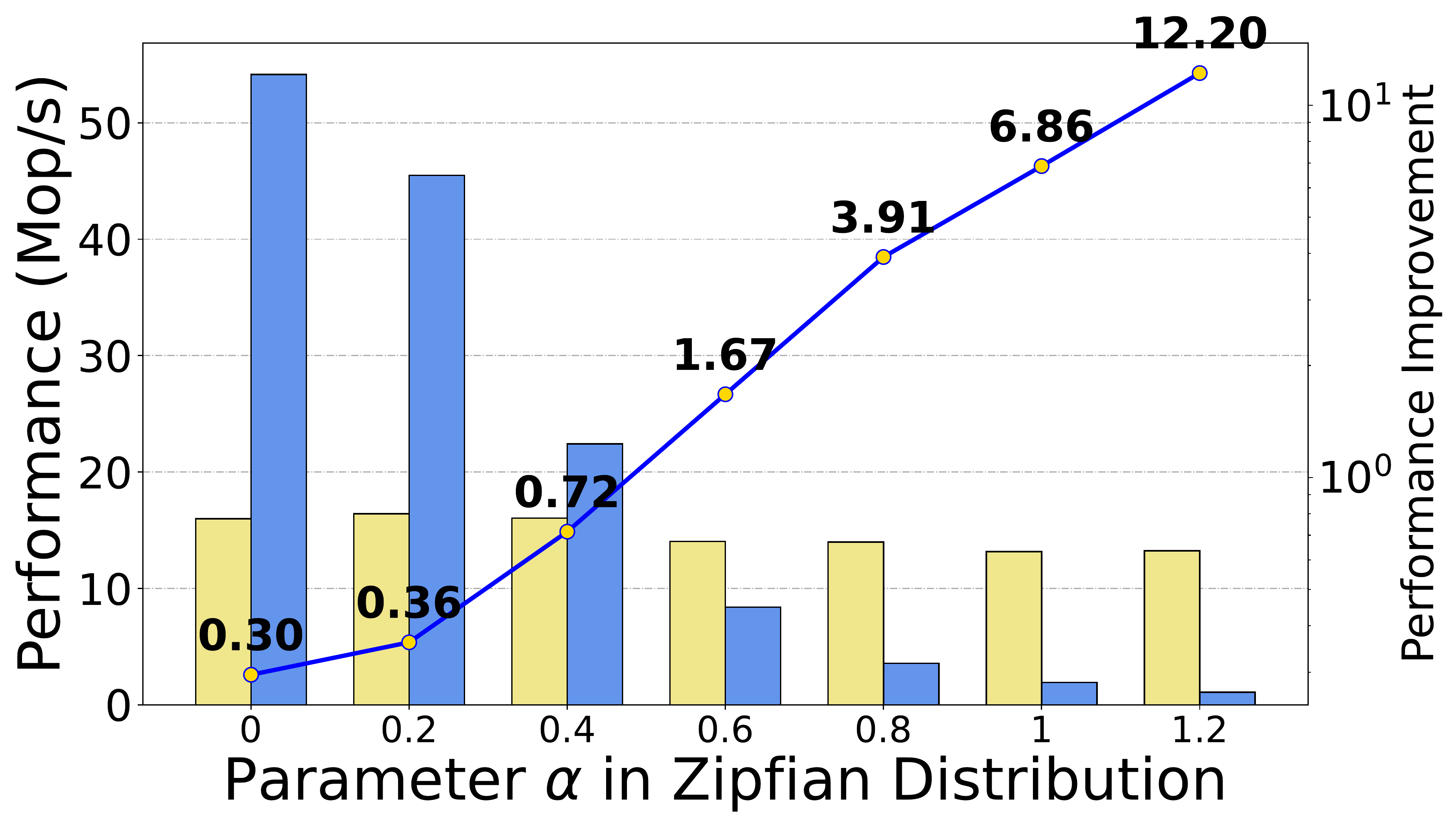}}
                \hspace{0.05in}
  \subcaptionbox*{}
    {\makebox[0.3\linewidth]{\includegraphics[width=0.2\linewidth]{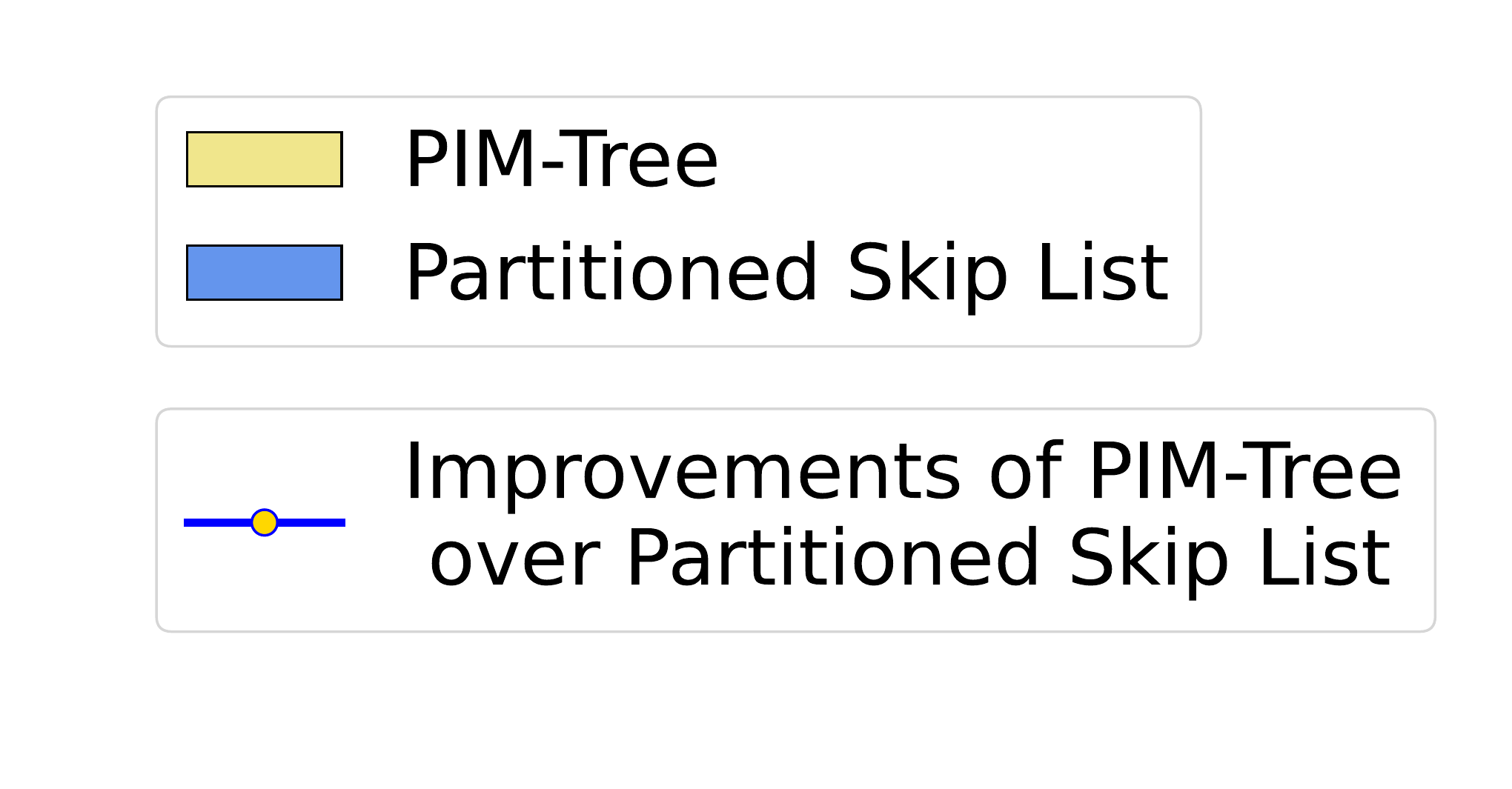}}}
  \caption{Throughput performance of ordered index operations. All
  operations other than \SCAN are run using a batch size of $10^6$;
  \SCAN uses a batch size
  of $10^4$, with 100 elements retrieved by each \SCAN operation in
  expectation.}
  \label{fig:time_performance_balance}
\end{figure*}

\begin{figure}[t]
  \centering
  \includegraphics[width=0.9\linewidth]{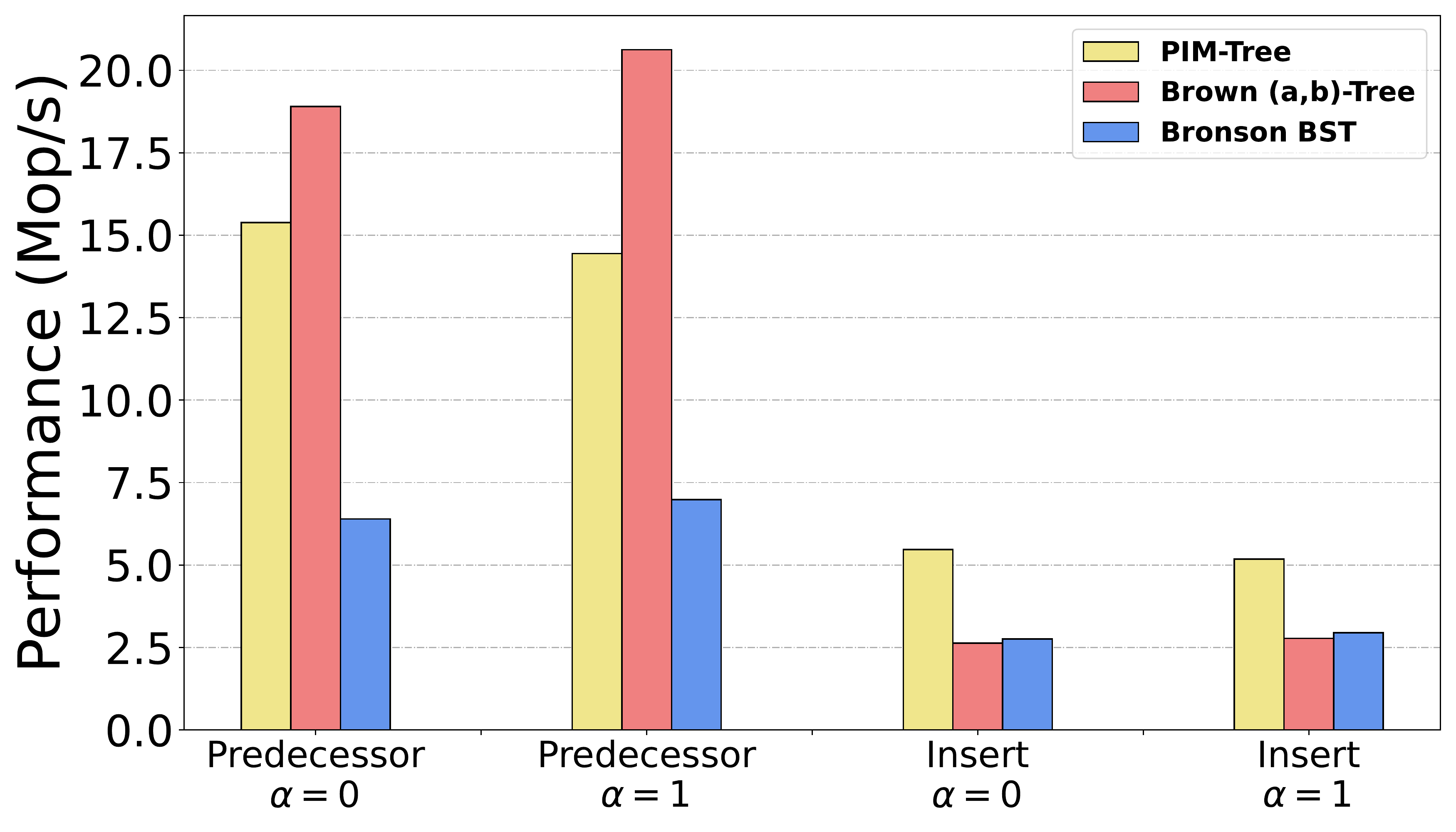}
  \caption{Throughput of PIM-tree versus SOTA traditional indexes.}
\label{fig:throughput_traditional}
\end{figure}


\myparagraph{Performance.}
\Cref{fig:time_performance_balance} illustrates the throughput of the
range-partitioned skip list and the PIM-tree on microbenchmarks.
The performance of the range-partitioned baseline drops drastically as the query skew increases,
while the PIM-tree shows robust resistance to query skew. In fact, across all
operations, it is observed that PIM-tree is essentially unaffected by data
skew, obtaining similar running times for $\alpha=0$ and $\alpha=1.2$.
For $\alpha = 1.2$, PIM-tree outperforms the range-partitioned
baseline by 3.87--59.1$\times$.

It is observed that \GET operations are significantly simpler and achieve
higher throughput since a hash table is used as a shortcut
(the same holds for \UPDATE operations), whereas
\PREDECESSOR and \SCAN operations must go through the entire ordered index.
In \Cref{fig:time_performance_balance}(c), \INSERT on the
range-partitioned baseline crashes when $\alpha = 1.2$, because
skewed \INSERT causes imbalanced data placement over PIM modules, then
causes overflow of local memory on some PIM modules.
Although this problem could be solved by a rebalancing scheme, the rebalancing
process itself will cause load imbalance as it requires sending data from
the overflowing PIM modules to other less-loaded PIM modules.
It is observed that even if this improvement to the baseline (with which
the existing range-partitioning solutions in the literature are not
equipped) were to be made, the throughput of PIM-tree would still be
significantly larger than range-partitioning, as this would be extra
work that the baseline must perform during the execution.

\revision{
\Cref{fig:throughput_traditional} shows the performance of PIM-tree
compared with \stoa \bst ~\cite{DBLP:conf/ppopp/BronsonCCO10} and \abtree ~\cite{brown2017thesis} under our workload.
PIM-tree outperforms
traditional indexes in all test cases, except
throughput of \PREDECESSOR compared with \abtree.
}

\myparagraph{Execution breakdown}.
\label{sec:breakdown}
\begin{figure}[t]
    \centering
    \vspace{-0.6em}
    \includegraphics[width=\linewidth]{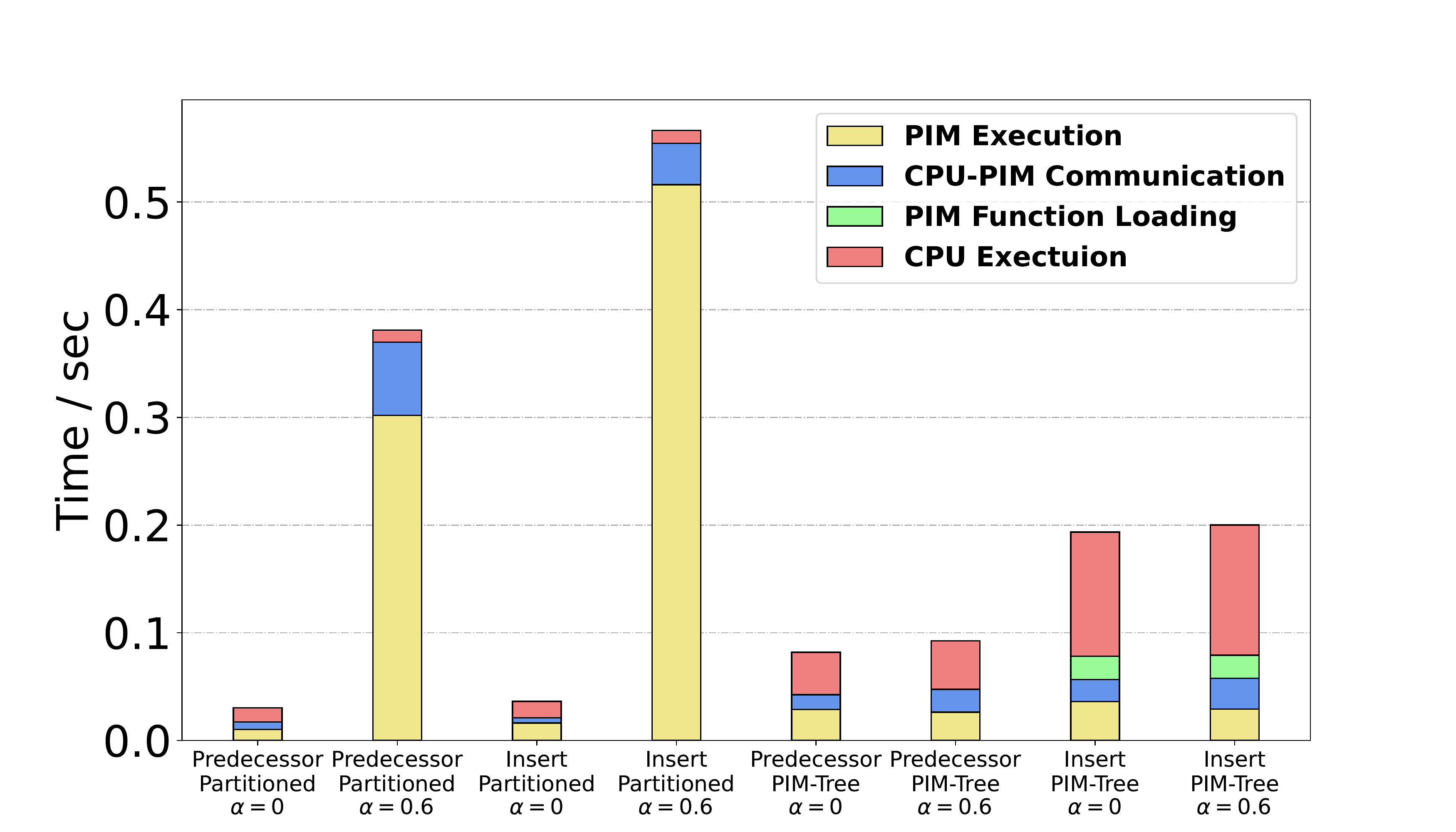}
    \caption{Time spent on each component of the program (without pipelining).}
\label{fig:time_breakdown}
\end{figure}
\Cref{fig:time_breakdown} shows the percentage of time spent in each component mentioned in
\Cref{sec:pipelining}. These results are derived with our pipelining optimization turned off, because pipelining would cause an overlapping of different components.
We select as typical examples the throughput of \PREDECESSOR and \INSERT, on range-partitioned skip list and the PIM-tree,
for uniform random workload and Zipfian-skewed workload with $\alpha = 0.6$.
Similar results also exist in the cases of other $\alpha$ values.


For range-partitioned skip list,
\textit{PIM Execution} and \textit{CPU-PIM Communication}
dominates the time cost of skewed workloads,
mainly because the bottleneck of the entire execution---the busiest PIM modules---
are receiving a growing number of tasks.

For uniform random workloads,
\textit{PIM Execution} only takes a small proportion of the total time cost,
though almost all comparisons are executed in PIM modules.
It is inferred that parallelism is fully exploited when a large number of PIM modules are involved in this case.
We believe that this implies that PIM-based systems are an ideal platform for parallel index structures.

PIM-tree \INSERT spends time loading PIM program modules during
execution, as the full program size exceeds the current size of
instruction memory on PIMs intended to store the PIM program.  This
limit is discussed in \Cref{sec:pim_program}.

\myparagraph{Effect of Optimizations.}
\begin{figure}[t]
  \centering
  \includegraphics[width=0.9\linewidth]{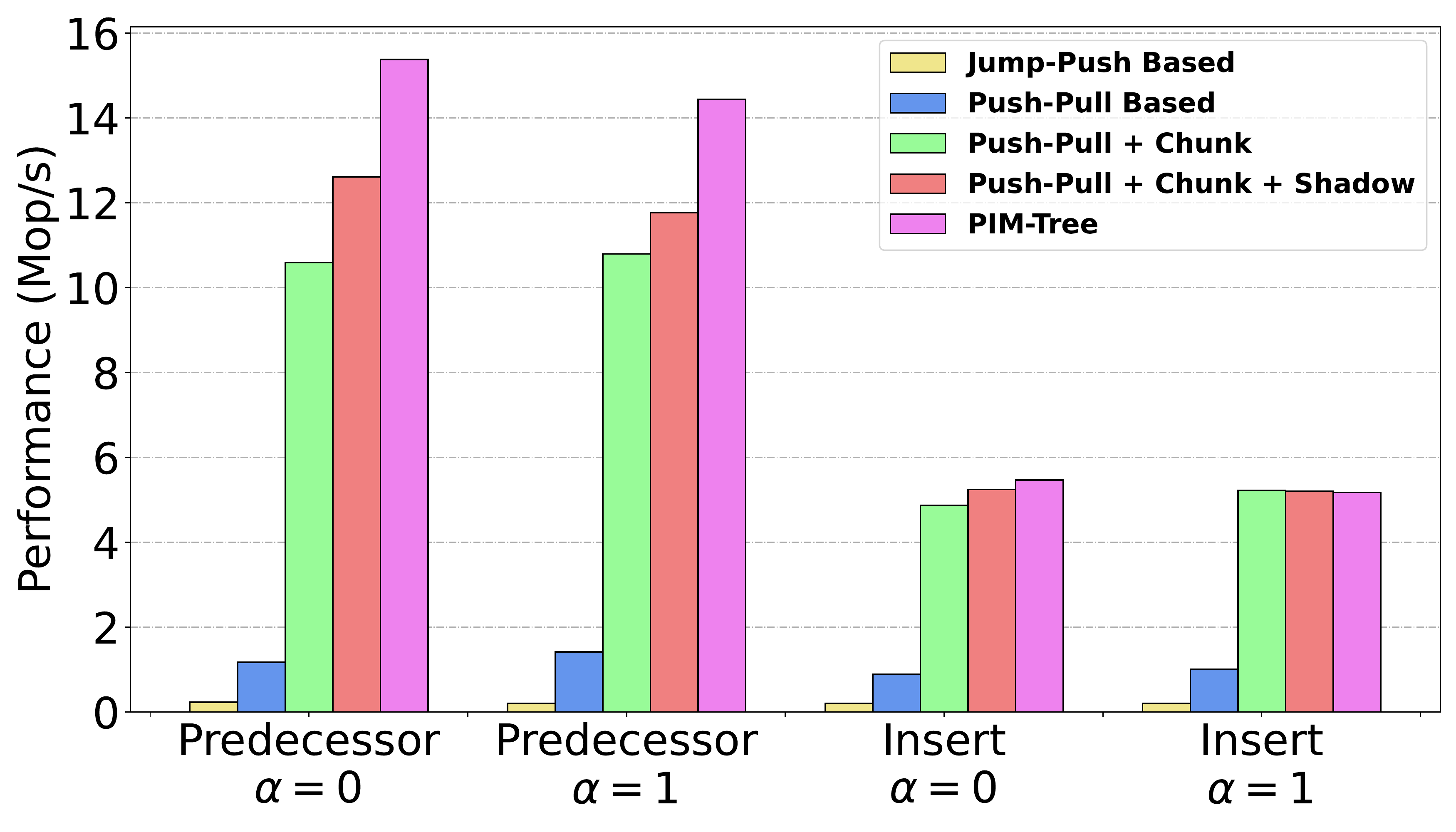}
  \caption{Impact of three optimizations on final performance.}
\label{fig:pim_tree_optimizations}
\end{figure}
\Cref{fig:pim_tree_optimizations} shows the impact of different
optimizations on our ordered index.  Here, we start with our Jump-Push
baseline (the PIM-balance skip list in~\cite{kang2021processing}).
Replacing Jump-Push with \PushPull provides up to 6.8$\times$ throughput
improvements across all test cases.  Adding Chunking provides the
biggest improvement jump, up to 9.0$\times$, across all test cases, while adding shadow
subtrees mostly benefits \PREDECESSOR under no skew. (\INSERT get minor
benefits because it needs to maintain this supplementary data structure.)
Finally, adding pipelining---thereby implementing the complete
PIM-tree algorithm---provides additional benefit for
\PREDECESSOR. (Pipelining is not implemented for \INSERT because it
would require interleaved \INSERT batches, which is not supported in
our implementation.)
Compared to the Jump-Push baseline, PIM-trees are up to 69.7$\times$ higher throughput
for the settings studied.

\begin{figure}[t]
  \centering
  \includegraphics[width=0.9\linewidth]{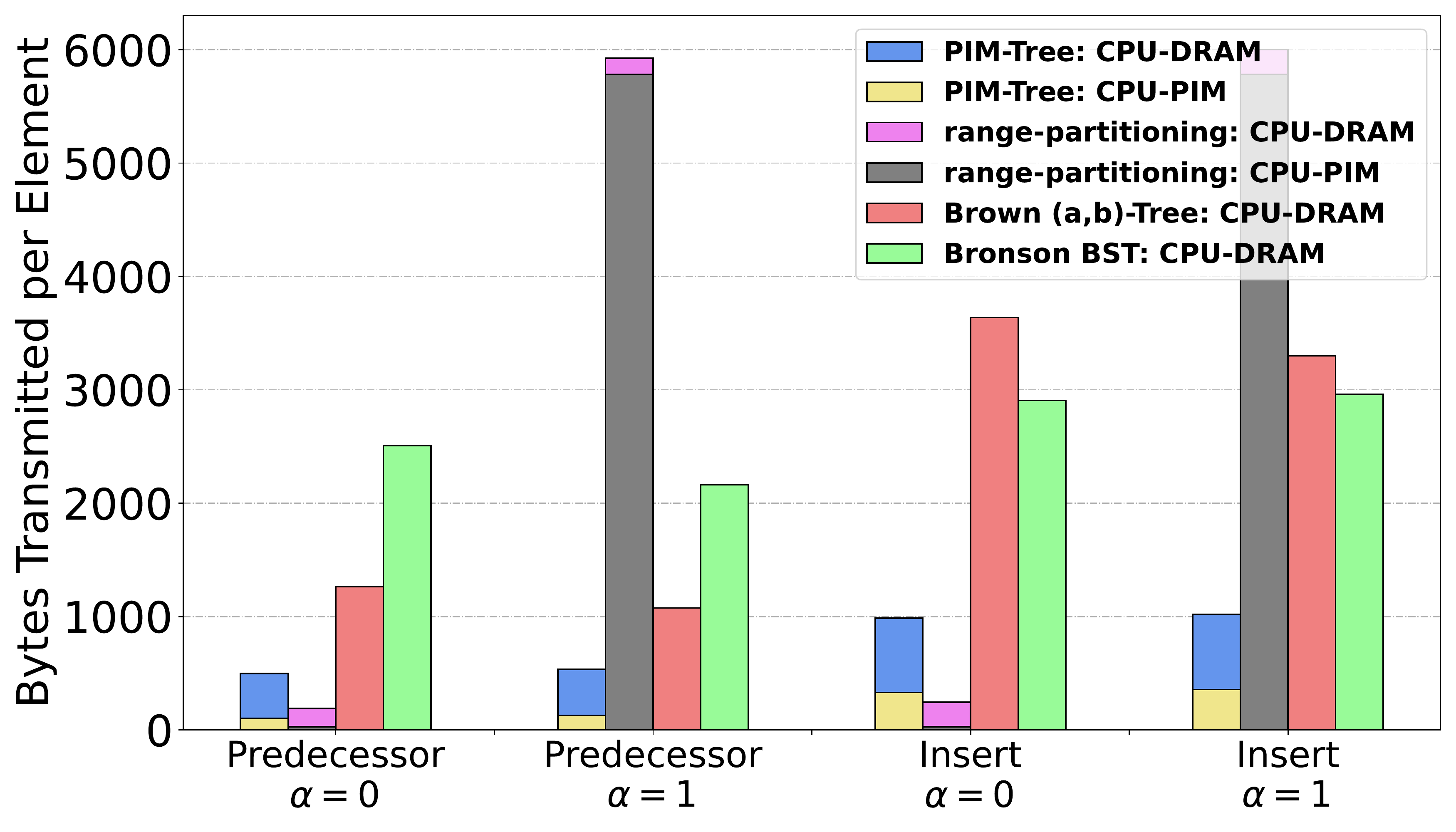}
  \caption{Average communication on the memory bus per operation in bytes.}
\label{fig:communication}
\end{figure}

\begin{figure}[t]
\centering
\includegraphics[width=0.9\linewidth]{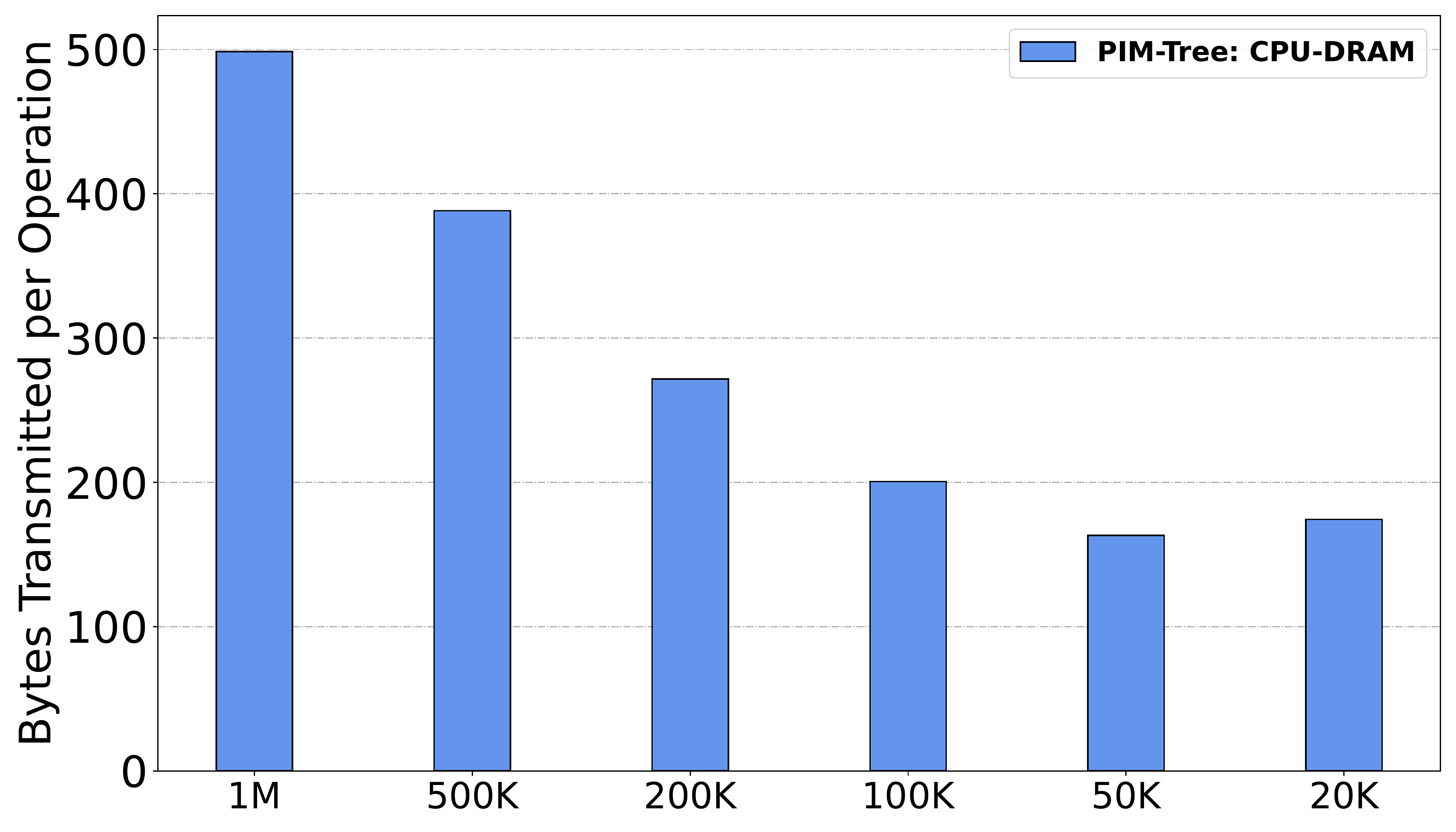}
\caption{Average CPU-DRAM communication per predecessor operation for PIM-Tree with different batch size in bytes.}
\label{fig:communication_of_different_batch_size}
\end{figure}

\revision{
\myparagraph{Memory bus communication.}
\Cref{fig:communication} shows the average amount of communication
for PIM-tree, range-partitioned skip list, and traditional non-PIM indexes.
PIM-tree needs less communication than all traditional indexes.
Range partitioned skip lists outperform all competitors by much under
uniform random workload, but perform much worse in skewed workloads.

Another observation is that, while the PIM-tree stores all the data and
does most comparisons in PIM modules, most memory bus traffic
is between CPU and the DRAM.
This is because though PIM-tree algorithms requires $O(S)$
CPU-side memory for a batch of $S$ operations, the available setup with $11$ MB cache
is too small for batches of one million operations. As the result,
CPU side data overflow to DRAM and cause significant CPU-DRAM communication.
To show the effect of this overflow, in
\Cref{fig:communication_of_different_batch_size} we study the CPU side
communication as we run the 100 million uniform random
predecessor operations with different batch sizes.
Results show that the CPU-DRAM
communication is reduced by $67\%$ as we reduce batch size from 1M to 50K. We cannot
directly use smaller batch size because of the load balance requirements,
but this result hints that we can get much less CPU-DRAM communication
when running the PIM-tree on a machine with larger cache size.
}

\conffulldifferent{
  \revision{
  \myparagraph{\PushPull threshold choice}
  We study the PIM-tree \PREDECESSOR performance under different \PushPull
  threshold.
  In our microbenchmark with $\mathbf{\alpha} = 1$, we find that
  choosing a lower threshold leads to about $10\%$ throughput drop and up to $28\%$ more CPU-PIM
  communication. A higher threshold brings minor performance increase. Please refer to the full paper~\cite{pimtree2022full} for
  more details.
  }
}{
  \begin{figure}[t]
    \centering
    \includegraphics[width=0.8\linewidth]{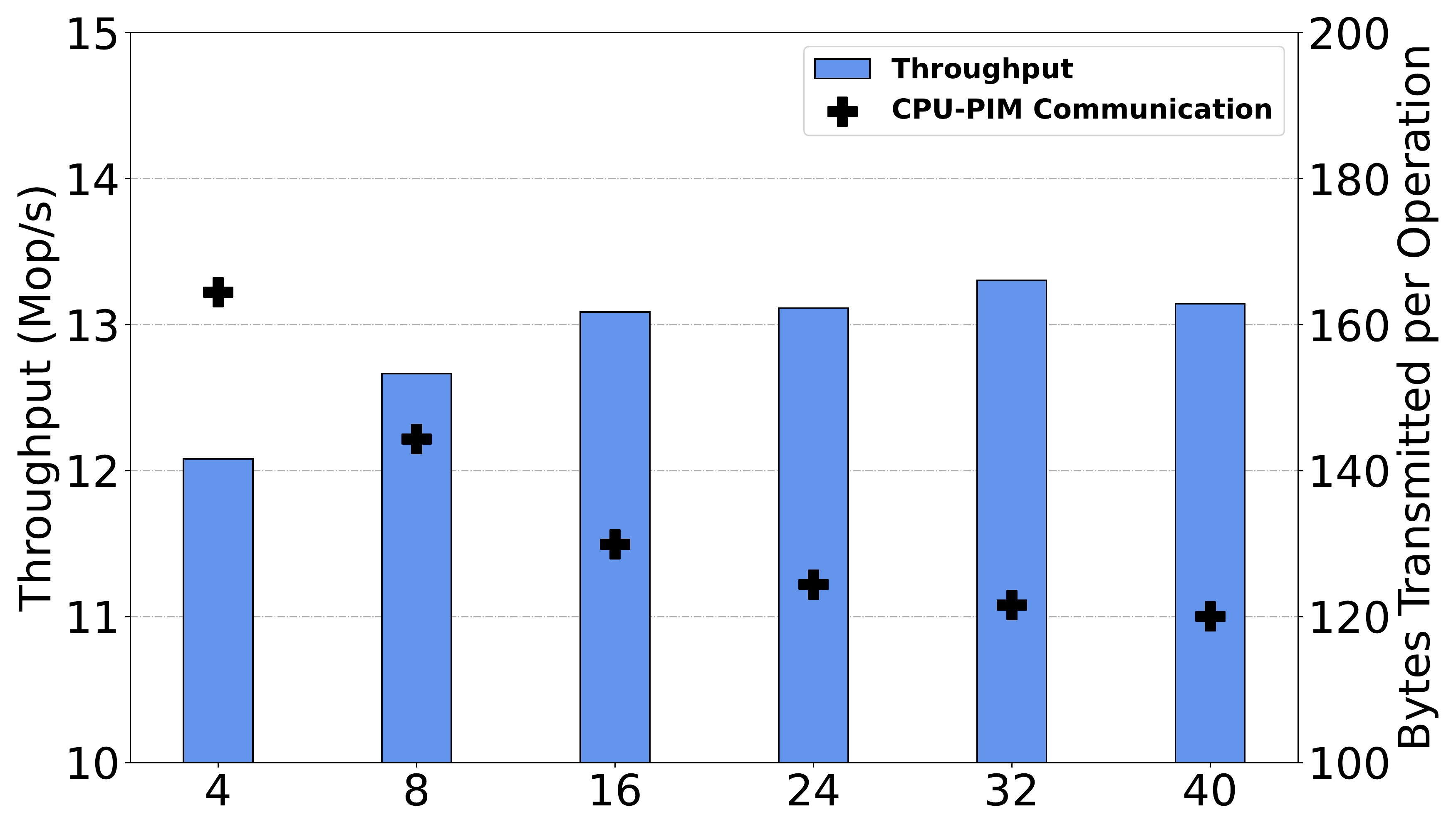}
    \caption{PIM-tree \PREDECESSOR performance under different \PushPull threshold factor $k$ with $\mathbf{\alpha} = 1$.}
  \label{fig:threshold_analysis}
  \end{figure}
  \myparagraph{\PushPull threshold choice} \Cref{fig:threshold_analysis} study the PIM-tree \PREDECESSOR performance
  under different \PushPull threshold factor $k$.
  Recall that we use different threshold in L1 and L2:
  the threshold of the \PushPull rounds in L1 is $K = k$, and the threshold of the \PullOnly round in L2 is $K = H'_{L2} * k$.
  We run this experiment on the \PREDECESSOR workload of our microbenchmark with $\mathbf{\alpha} = 1$.
  Results show that choosing a lower threshold leads to about $10\%$ throughput drop and up to $28\%$ more CPU-PIM
  communication. A threshold factor higher than $B$ brings minor performance increase.

  We believe that this improvement comes from the chunked skip list structure. If the size of each node is strictly
  $B = 16$, choosing threshold factor $k = 16$ will be optimal. However, the chunked skip list nodes could have larger sizes.
  According to the geometric distribution, there will be large nodes with more than $200$ keys when the PIM-tree has millions of nodes.
  Larger nodes have both higher \Pull cost and probability: they covers a larger key range, which collects more queries.
  Meanwhile, smaller nodes have both lower \Pull cost and probability. Increasing the threshold factor can help alleviate
  this effect.
}

\conffulldifferent{
  \myparagraph{Energy Evaluation}.
  \revision{
  \pimtree{} costs roughly
  5$\times$--10$\times$ less energy on skewed cases,
  compared to the range-partitioned baseline on PIM.
  Please refer to the full version of this paper~\cite{pimtree2022full} for details.
  }
}{
  \begin{figure}
    \centering
    \subcaptionbox{100M \INSERT operations\label{fig:insert_energy}}
      {\includegraphics[width=1\linewidth]{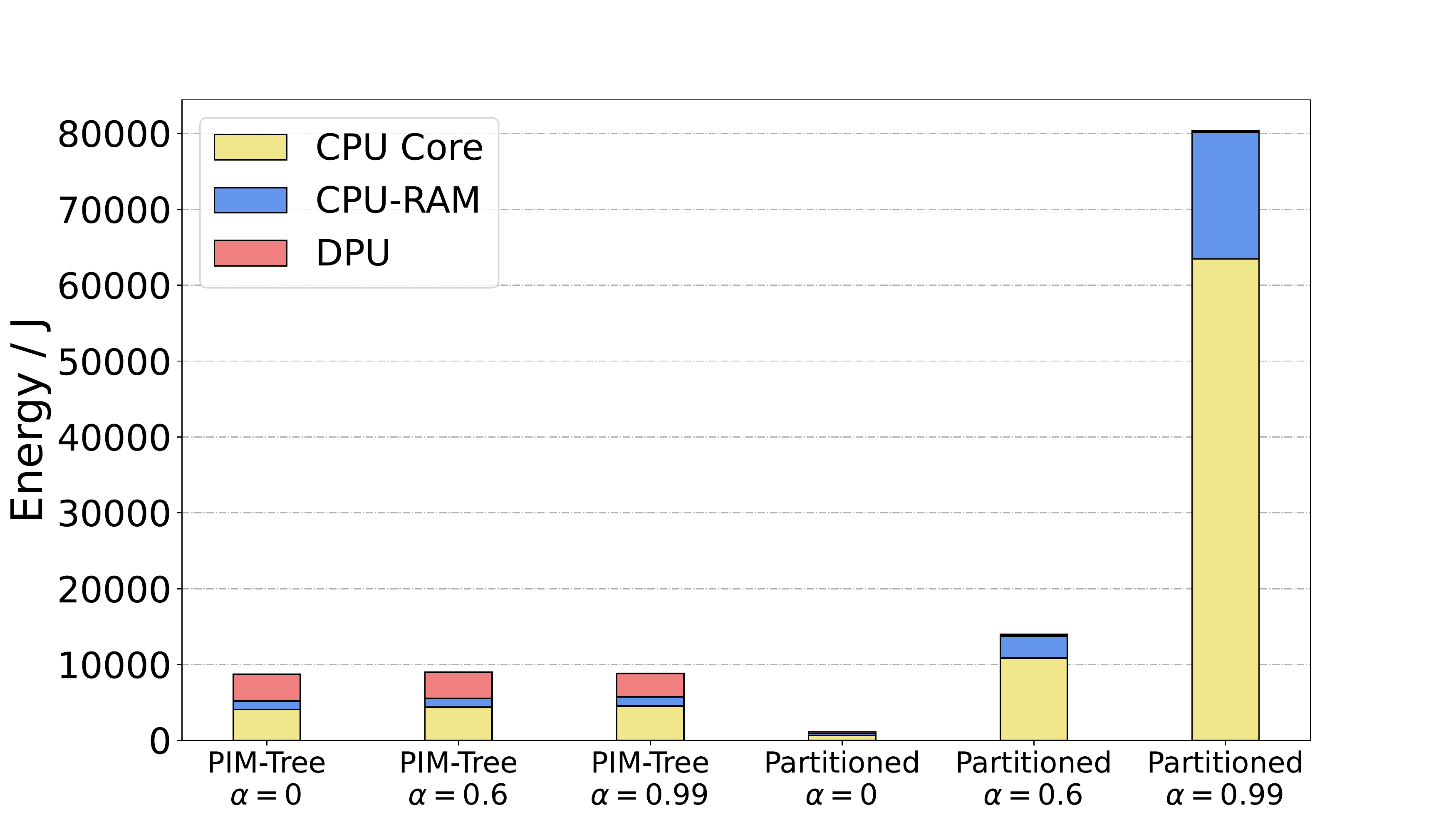}}
    \subcaptionbox{100M \PREDECESSOR operations\label{fig:search_energy}}
     {\includegraphics[width=1\linewidth]{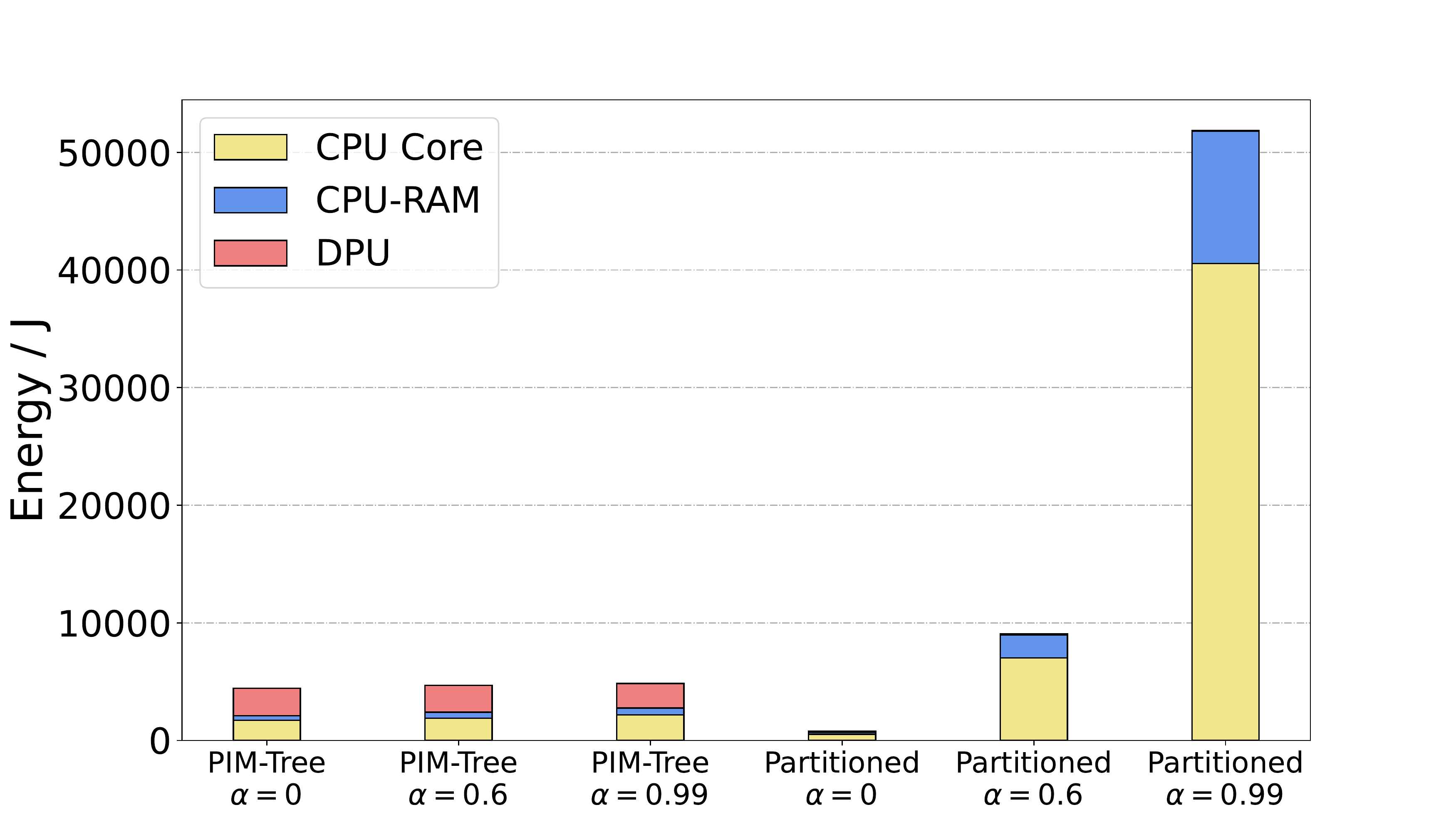}}
    \subcaptionbox{1M \SCAN operations with size 100\label{fig:scan_energy}}
      {\includegraphics[width=1\linewidth]{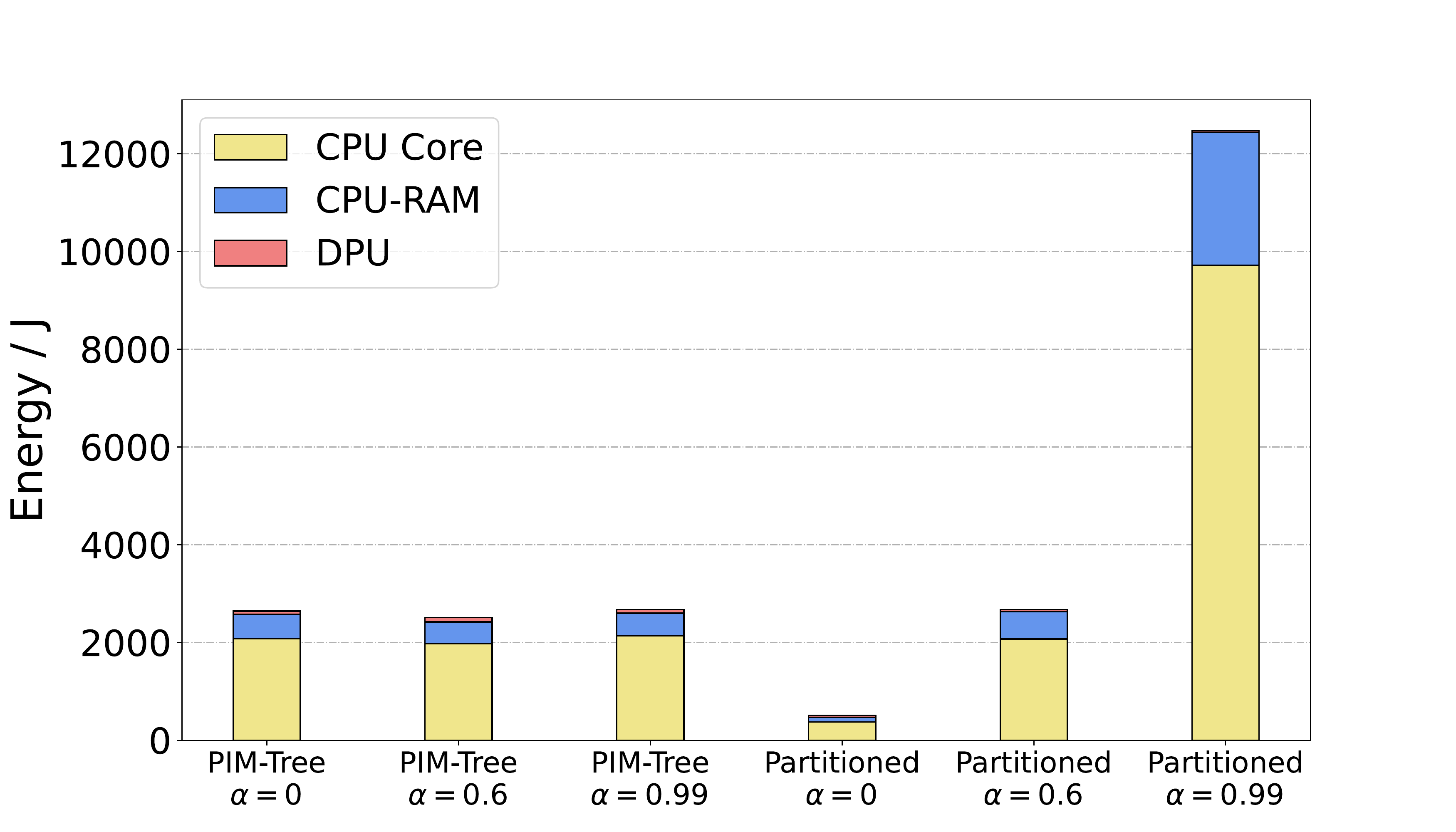}}
    \caption{Energy Cost of operations on PIM-tree and range-partitioned skip list.}
    \label{fig:energy}
  \end{figure}

  \myparagraph{Energy Evaluation}.  We also evaluate the energy
  consumption of PIM-tree versus the range-partitioned baseline.  The
  energy evaluation is carried out separately for the CPU host and the
  PIM modules.  Intel RAPL~\cite{khan2018rapl} is used to collect CPU
  energy consumption statistics.
  Meanwhile, the number of instructions, cycles,
  WRAM accesses, and MRAM accesses are collected on the UPMEM PIM
  modules (called \textit{DPUs}).  WRAM and MRAM are the two types of
  memories on each DPU.  We estimate the DPU energy consumption by
  multiplying these values by the hardware-related energy weights
  provided by UPMEM~\cite{upmem}, and summing the result.

  \Cref{fig:energy} illustrates the energy consumption of PIM-tree and
  range-partitioned skip list, on the critical operations of
  \INSERT, \PREDECESSOR and \SCAN with the same element size per batch on a same dataset.
  The figure shows an energy break down into three components:
  \textit{CPU Core}, \textit{CPU-RAM} (i.e., CPU-PIM communication and CPU-DRAM communication), and
  \textit{DPU} (i.e., executing PIM programs).
  Overall, \pimtree{} sacrifices roughly 1.5$\times$--2$\times$
  energy consumption on unskewed cases, in return for roughly
  5$\times$--10$\times$ energy reduction on skewed cases.

  Two additional findings can be drawn from this evaluation. First, DPU
  energy consumption is relatively stable against skew in all designs.
  This is because, under the assumption of DPU energy evaluation provided
  by UPMEM, an energy-efficient DPU can be turned on only
  when it is called for a task and turned off as soon as it returns the
  results to the CPU side.  Therefore, the DPU energy consumption is
  positively correlated only to the number of executed operations.
  Meanwhile, even in the skewed cases, the total number of operations
  required to execute is approximately the work of these parallel
  operation batches and remains relatively constant.

  Second, the energy consumption of the CPU and CPU-PIM communication
  is highly sensitive to skew in the range-partitioned baseline, while robust
  to skew in PIM-tree.  We found on microbenchmarks that CPU energy
  consumption is strictly linearly proportional to the operating time,
  regardless of PIM data structure designs and operation batch types.
  One possible explanation is that all these designs assure that the CPU
  runs on its maximum capacity of processing and data communication, so
  both required time and energy consumption are proportional to the
  number of required CPU operations.  Under such explanation, we suggest
  that the energy inefficiency on CPU of the range-partitioned baseline
  under skewed cases is equivalent to the time inefficiency.

  We argue from the above analysis that the PIM-tree is a highly
  energy-efficient design in skewed cases, with the cost of only a
  little more energy consumption in unskewed cases (used for more
  complicated data structure construction).
}

\conffulldifferent{
  \myparagraph{YCSB workload}.
  \revision{
    \pimtree{} achieves roughly
    9.5$\times$--32$\times$ higher throughput on skewed cases,
    compared to the range-partitioned baseline on PIM.
    Please see the full version of this paper~\cite{pimtree2022full} for more details.
  }
}{
  \subsection{YCSB Workload}
  \begin{figure}[t]
    \centering
    \includegraphics[width=0.8\linewidth]{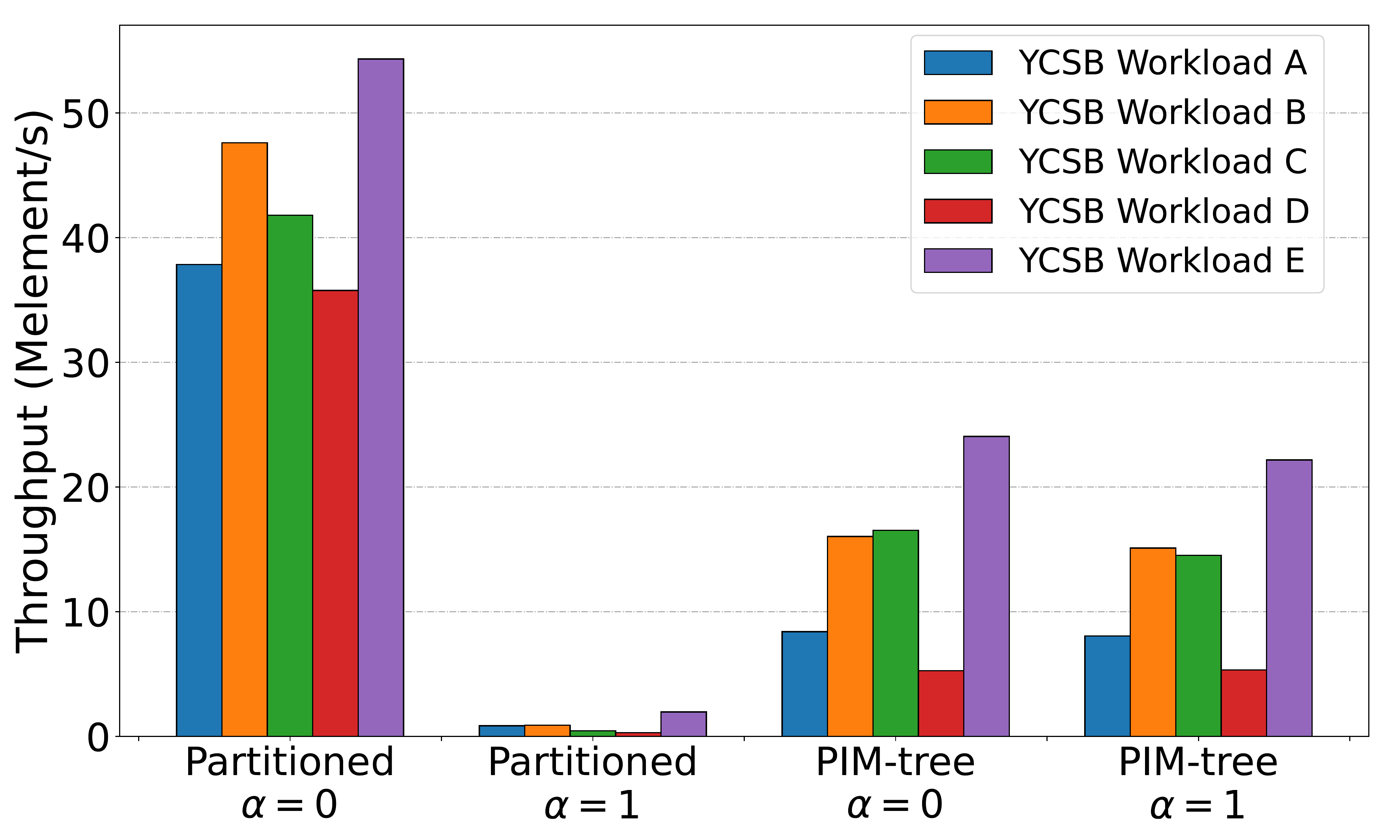}
    \caption{YCSB workload throughput.}
  \label{fig:ycsb_throughput}
  \end{figure}

  \revision{
  Finally, we test our indexes using five YCSB-like \cite{cooper2010benchmarking} workloads:
  \begin{itemize}
    \item[\textbf{A}] write-intensive (50\% \PREDECESSOR and 50\% \INSERT)
    \item[\textbf{B}] read-intensive (95\% \PREDECESSOR and 5\% \INSERT)
    \item[\textbf{C}] \PREDECESSOR-only
    \item[\textbf{D}] \INSERT-only
    \item[\textbf{E}] short-ranges (95\% \SCAN and 5\% \INSERT)
  \end{itemize}
  }
  We use the method in \S\ref{sec:micro_setup} to generate zipfian-skewed
  operations with $\alpha = 0, 1$.  All workloads are warmed up with
  \dataSize elements. Workloads A--D are tested with \testSize
  operations. Workload E is tested with 20 million operations, because
  each \SCAN range is expected to return 100 elements, so the total
  number of tested elements is enough for the analysis.
  For workloads with mixed-type operations (A, B, E), we load operations into separate same-type
  batches, and run a batch atomically when its expected returned size
  exceeds \batchSize \cite{sun2019supporting}.

  The results are shown in \Cref{fig:ycsb_throughput}, and again show
  the fragility of the range-partitioned skip list and the robustness of
  PIM-tree under skewed workloads.
}

\subsection{\revision{Workload of Real-world Skewness}}

\begin{figure}[t]
  \centering
  \includegraphics[width=0.7\linewidth]{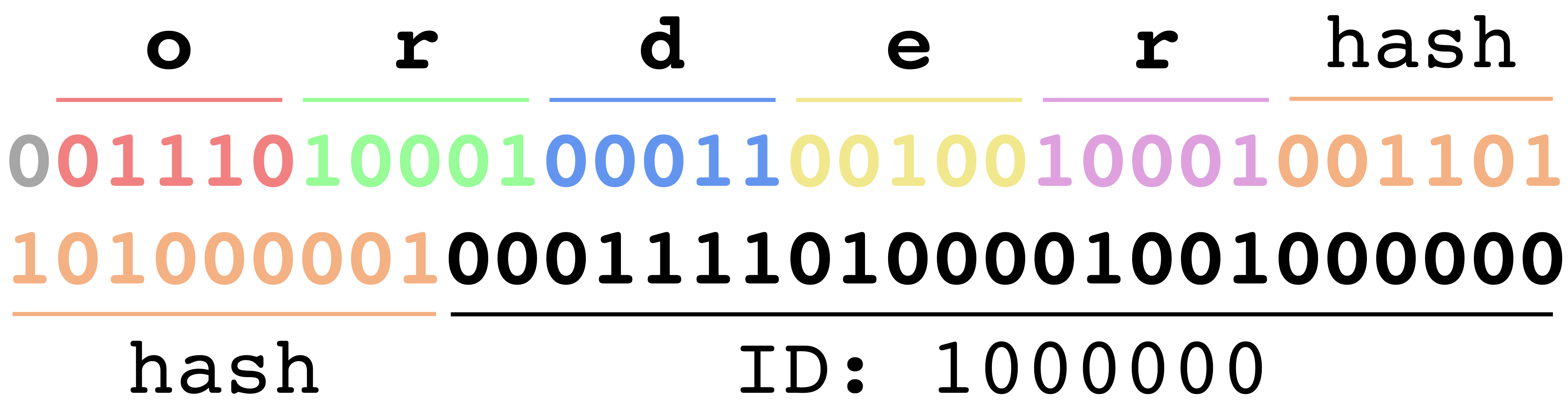}
  \caption{An example that convert word ``ordered'' in document id ``1000000''
  to a 64-bit integer.}
\label{fig:wikipedia_translate}
\end{figure}

\begin{figure}[t]
  \centering
  \includegraphics[width=0.8\linewidth]{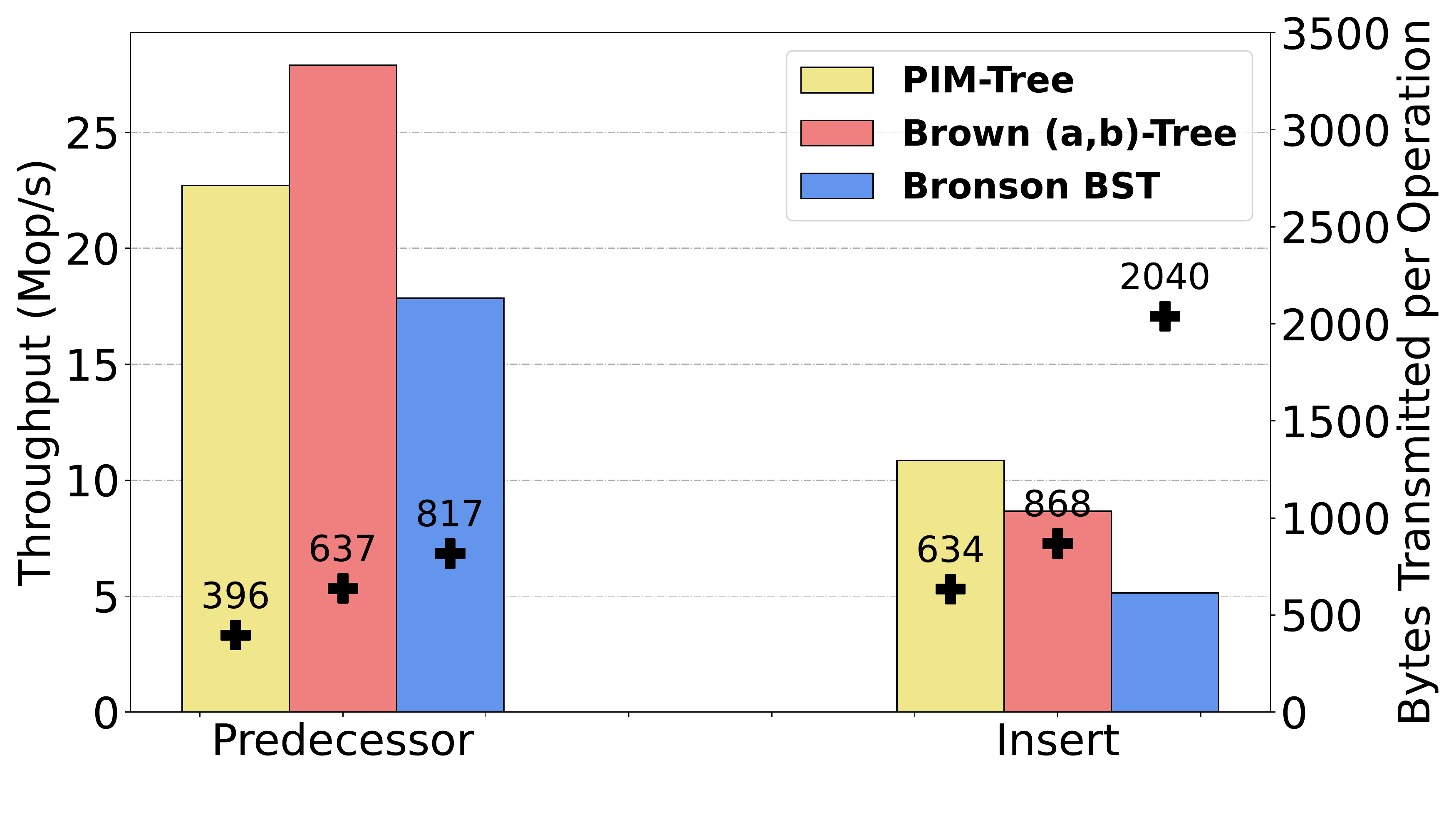}
  \caption{Throughput on the wikipedia workload.}
\label{fig:wikipedia_throughput}
\end{figure}

\revision{
In this section, we test the PIM-tree over a workload with real-world skewness
using the publicly available wikipedia dataset \cite{wikipedia_database}, which
is a collection of documents from wikipedia.
To use this dataset in our test framework, we need to transform it
into a collection of 8 byte key-value pairs, then run operations over them.
To be specific, we first extract words from each document,
lowercase them, then use
(word, document id) pairs as keys, and a random 8 byte integer as values.
Because our indexes only support 8 byte integer keys,
we need to transform the (word, document id) pairs into 8 byte integers.
The transformation is shown in \Cref{fig:wikipedia_translate}.
We use 40 bits to represent the word, and 23 bits to store the document id
(as there's less than $2^{23}$ documents). In the word part, we use 5 bits
for each of the first 5 letters, then store the hash value of the whole
word in the following 15 bits to avoid collision.
After this transformation, the generated integers preserve the two skewness
of english words: (i) word frequency skewness
(some words are used more than others)
and (ii) word distribution skewness in the
dictionary order space (words with some prefix are used more).

In this test, we pick the first 1.2 billion keys:
the first 1 billion words used for initialization, and the following
200 million used for evaluation.
These keys covers the first 5.1 million documents,
which is $63\%$ of the whole dataset.
There are 3.9 million
unique words, and pairing the word and document id generates 529 million unique keys.
We get duplicated keys only for the same word in the same document.
Because the duplication rate of keys is about $2X$, we also double the batch size
of the PIM-tree to two million.

The result is shown in \Cref{fig:wikipedia_throughput}, where the
throughputs are shown as the bars, and communication (bytes
transmitted per-operation) as labeled points.
All indexes experience higher throughput and lower communication in
this workload than in microbenchmarks
because of the replicated keys. \laxman{higher and lower than
what?}
Comparing different indexes gives results similar to that of microbenchmarks:
PIM-tree has lower predecessor throughput than the (a,b)-tree, but
outperforms traditional indexes in all other metrics.
}

\revision{
\section{Discussion}
\pimtree{} outperforms conventional indexes in throughput in most cases, but very occasionally cannot win, e.g. only in \PREDECESSOR compared with \abtree in our paper.
We address here three hardware limits of the current PIM system by way
of explanation, and to describe future changes to the hardware that
would result in even better performance for PIM-optimized data
structures.

The first factor is the limited CPU-PIM bandwidth on UPMEM's newly developed hardware. \hongbo{todo: update numbers. we can keep the story}
When carrying out a $50\%$ read - $50\%$ write task,
the bandwidth obtained on UPMEM machine is
16GB/s, $1.9\times$ slower than the shared-memory machine we use with a bandwidth of 31GB/s on the same workload.
Even under such significant bandwidth limitations, \pimtree{} still
achieves better or comparable performance to DRAM-only
indexes,
primarily because it greatly reduces inter-module communication.
Designing hardware to improve CPU-PIM bandwidth is thus an important
direction, and one that we expect improvements for in the future.
Therefore, we believe that \pimtree{} will outperform conventional
indexes in all cases in terms of throughput in the future.

Another issue is that the limited size of PIM program prevent us from
more complicated designs. Current workaround, the dynamic program loading
is too costly. We believe this problem will be solved in future hardware
by a larger instruction memory, in other ways.

The last limit is that of inadequate CPU cache, as mentioned in Section \ref{sec:micro_setup}.
CPU-DRAM communication caused by cache overflow makes most
of memory bus communication, and this can be alleviated by a larger
cache. We believe an adequate cache will be important in future
PIM systems.
}

\section{Conclusion}

This paper presented \textit{PIM-tree}, the first ordered index for
PIM systems that achieves both low communication and high load balance
in the presence of data and query skew.  We presented the first
experimental evaluation of ordered indexes on a real PIM system,
demonstrating up to $69.7\times$ and $59.1\times$ higher throughput
than the two best prior PIM-based methods and down to $0.4\times$ less communication than two \stoa conventional indexes.
Key ideas include \textit{push-pull
  search} and \textit{shadow subtrees}---techniques likely to be
useful for other applications on PIM systems due to their
effectiveness in reducing communication costs and managing skew.  Our
future work will explore such applications (e.g., radix-based indexes, graph analytics).

\begin{acks}
  We thank R\'{e}my Cimadomo, Julien Legriel, Damien Lagneux,
  \revision{
  \conffulldifferent{etc.}{and all the other folks}
  }at UPMEM for providing extensive access to
  \revision{\conffulldifferent{their}{the UPMEM}}
  system and help
  whenever needed. This work would not have been possible without
  \revision{\conffulldifferent{their support.}{this level of support.}}
  This research was supported by NSF grants
  CCF-1910030, CCF-1919223, CCF-2028949, and CCF-2103483, VMware
  University Research Fund Award, Parallel Data Lab \conffulldifferent{}{(PDL)} Consortium
  \ifx\confversion\undefined
  (Alibaba, Amazon, Datrium, Facebook, Google,
  Hewlett-Packard Enterprise, Hitachi, IBM, Intel, Microsoft, NetApp,
  Oracle, Salesforce, Samsung, Seagate, and TwoSigma)
  \fi
  and National Key
  \revision{\conffulldifferent{R\&D}{Research \& Development}}
  Program of China (2020YFC1522702).
\end{acks}

\bibliographystyle{ACM-Reference-Format}
\bibliography{ref}

\end{document}